\documentclass[12pt,a4paper,twoside]{amsart} 
\usepackage{latexsym}
\usepackage{amsmath}
\usepackage{amssymb}
\usepackage{amsfonts}
\usepackage{ifthen}
\usepackage{mathrsfs}               
\usepackage{bbm}                    
\usepackage{bbold}                  
\usepackage{textcomp}               
\usepackage{amstext}
\usepackage{enumerate}
\usepackage{amstext}
\newcommand{\CC}{\mathbb{C}}

\newcommand{\NN}{\mathbb{N}}

\newcommand{\RR}{\mathbb{R}}

\newcommand{\ZZ}{\mathbb{Z}}
\newcommand{\supp}{\mathrm{supp}}
\newcommand{\const}{\mathrm{const}}
\newcommand{\Ran}{\mathrm{Ran}}
\newcommand{\nr}{\mathrm{nr}}
\newcommand{\np}{\mathrm{np}}
\newcommand{\sr}{\mathrm{sr}}

\newcommand{\id}{\mathbbm{1}}
\newcommand{\klg}{\leqslant} 
\newcommand{\grg}{\geqslant}          
\newcommand{\ve}{\varepsilon}
\newcommand{\vp}{\varphi}

\newcommand{\vr}{\varrho}

\newcommand{\vs}{\varsigma}

\newcommand{\wt}[1]{\widetilde{#1}}
    
\newcommand{\SPn}[2]{\langle \,#1\,|\,#2\, \rangle} 
\newcommand{\SPb}[2]{\big\langle \,#1\,\big|\,#2\, \big\rangle} 
\newcommand{\SPB}[2]{\Big\langle \,#1\,\Big|\,#2\, \Big\rangle}
\newcommand{\ol}[1]{\overline{#1}} 

\newcommand{\mr}[1]{\mathring{#1}} 
\newcommand{\V}[1]{\mathbf{#1}}
\newcommand{\valpha}{\boldsymbol{\alpha}}

\newcommand{\vsigma}{\boldsymbol{\sigma}}
\newcommand{\vSigma}{\boldsymbol{\Sigma}}
\newcommand{\veps}{\boldsymbol{\varepsilon}}
\newcommand{\LO}{\mathscr{L}}      

\newcommand{\HR}{\mathscr{H}}

\newcommand{\Fock}{\mathscr{F}_{\mathrm{b}}}
\newcommand{\core}{\mathscr{D}}
\newcommand{\dom}{\mathcal{D}}
\newcommand{\form}{\mathcal{Q}}
\newcommand{\spec}{\mathrm{\sigma}}

\newcommand{\PAm}{P^-_\mathbf{A}}               
\newcommand{\PA}{P^+_{\mathbf{A}}}

\newcommand{\PAN}{P_{\V{A},N}^{+}}
\newcommand{\PANb}{{P}_{\V{A},N}^\bot}
\newcommand{\PANs}{{P}_{\V{A},N}^\sharp}
\newcommand{\PANf}{{P}_{\V{A},N}^\flat}
\newcommand{\Pa}[1]{P_\V{A}^{+,(#1)}}

\newcommand{\Paf}[1]{P_\V{A}^{\flat,(#1)}}
\newcommand{\SA}{S_{\mathbf{A}}}  
\newcommand{\DA}{D_{\mathbf{A}}}                
\newcommand{\DO}{D_{\mathbf{0}}}
\newcommand{\D}[1]{D_{#1}}
\newcommand{\RA}[1]{R_{\mathbf{A}}(#1)}
\newcommand{\Hf}{H_{\mathrm{f}}}                           
\newcommand{\HT}{\check{H}_\mathrm{f}}        
\newcommand{\VC}{V_\mathrm{C}}
\newcommand{\ad}{a^\dagger}                     
\newcommand{\UV}{\Lambda}             
\newcommand{\ball}[2]{\cB_{#1}(#2)}
\newcommand{\cA}{\mathcal{A}}
\newcommand{\cB}{\mathcal{B}}

\newcommand{\cS}{\mathcal{S}}
\newcommand{\cT}{\mathcal{T}}

\newcommand{\sC}{\mathscr{C}}
\newcommand{\sD}{\mathscr{D}} 

\newcommand{\sR}{\mathscr{R}}

\newcommand{\sK}{\mathscr{K}}

\newcommand{\sZ}{\mathscr{Z}} 
\newcommand{\fA}{\mathfrak{A}}

\renewcommand{\Re}{\mathrm{Re}\,}
\newtheorem{theorem}{Theorem}[section]
\newtheorem{lemma}[theorem]{Lemma}
\newtheorem{proposition}[theorem]{Proposition}
\newtheorem{corollary}[theorem]{Corollary}
\newtheorem{hypothesis}[theorem]{Hypothesis}
\theoremstyle{remark}

\newtheorem*{example}{Example}
\numberwithin{equation}{section}
\title[Higher order estimates]{On higher order estimates
in quantum electrodynamics}
\author{Oliver Matte}
\address{Oliver Matte
Institut f\"ur Mathematik\\
TU Clausthal\\
Erzstra{\ss}e 1\\
D-38678 Clausthal-Zellerfeld, Germany\\
{\em On leave from:} Mathematisches Institut\\
Ludwig-Maximilians-Universit\"at\\
Theresienstra{\ss}e 39\\
D-80333 M\"unchen, Germany.}
\email{matte@math.lmu.de}

\date{\today}
\begin{document}

\begin{abstract}
We propose a new method to derive certain higher
order estimates in quantum electrodynamics. Our method is
particularly convenient in the application to the non-local
semi-relativistic models of quantum electrodynamics as it 
avoids the use of iterated commutator expansions.
We re-derive higher order estimates obtained earlier by
Fr\"ohlich, Griesemer, and Schlein and prove new estimates
for a non-local molecular no-pair operator.
\end{abstract}

\maketitle


\section{Introduction}

\noindent
The main objective of this paper is to present a new method to derive higher
order estimates in quantum electrodynamics (QED) of the form
\begin{align}\label{hoe-intro1}
\big\|\,\Hf^{n/2}\,(H+C)^{-n/2}\,\big\|\,&\klg\,\const\,<\,\infty\,,
\\ \label{hoe-intro2}
\big\|\,[\Hf^{n/2}\,,\,H]\,(H+C)^{-n/2}\,\big\|\,&\klg\,
\const\,<\,\infty\,,
\end{align}
for all $n\in\NN$, where $C>0$ is sufficiently large.
In these bounds $\Hf$ denotes the radiation field energy of
the quantized photon field and $H$ is the full Hamiltonian
generating the time evolution of an interacting electron-photon
system. For instance, estimates of this type serve as one of
the main technical ingredients in the mathematical analysis
of Rayleigh scattering. In this context, \eqref{hoe-intro1} 
has been proven
by Fr\"ohlich et al. in the case where $H$ is
the non- or semi-relativistic Pauli-Fierz Hamiltonian
\cite{FGS2001}; a slightly weaker version of \eqref{hoe-intro2}
has been obtained in \cite{FGS2001} for all even values of $n$.
Higher order estimates of the form \eqref{hoe-intro1} 
also turn out to be useful in the study of the
existence of ground states in a no-pair model of QED
\cite{KMS2009b}. In fact, they imply that every eigenvector
of the Hamiltonian $H$ or spectral subspaces of $H$
corresponding to some bounded interval
are contained in the domains of
higher powers of $\Hf$. This information is very helpful in order 
to overcome numerous technical difficulties which are caused by
the non-locality of the no-pair operator. 
In these applications it is actually necessary to have
some control on the norms in \eqref{hoe-intro1} and
\eqref{hoe-intro2} when the operator $H$ gets modified.
To this end we shall give rough bounds on the
right hand sides of \eqref{hoe-intro1} and
\eqref{hoe-intro2} in terms of the ground state energy
and integrals involving
the form factor and the dispersion relation.

Various types of higher order estimates have actually been employed 
in the mathematical analysis of quantum field theories since a very long time.
Here we only mention the classical works \cite{GlimmJaffe1968,Rosen1971} on
$P(\phi)_2$ models and the more recent articles
\cite{DerezinskiGerard2000} again on a $P(\phi)_2$ model
and \cite{Ammari2000} on the Nelson model.

In what follows
we briefly describe the organization and the content of the present article.
In Section~\ref{sec-general} we develop the main idea behind
our techniques in a general setting. By the criterion established there
the proof of the higher order estimates
is essentially boiled down to the verification of certain form bounds
on the commutator between $H$ and a regularized version of $\Hf^{n/2}$.
After that, in Section~\ref{sec-comm}, we introduce
some of the most important operators appearing in QED and establish
some useful norm bounds on certain commutators
involving them. These commutator estimates provide the main ingredients
necessary to apply the general criterion of Section~\ref{sec-general} to
the QED models treated in this article.
Their derivation is essentially based on the pull-through formula
which is always employed either way to derive higher order
estimates in quantum field theories 
\cite{Ammari2000,DerezinskiGerard2000,FGS2001,GlimmJaffe1968,Rosen1971};
compare Lemma~\ref{le-clara1} below.
In Sections~\ref{sec-nr}, \ref{sec-PF}, and~\ref{sec-np}
the general strategy from Section~\ref{sec-general} is applied
to the non- and semi-relativistic Pauli-Fierz operators
and to the no-pair operator, respectively.
The latter operators are introduced in detail in these sections.
Apart from the fact that our estimate \eqref{hoe-intro2} is slightly
stronger than the corresponding one of \cite{FGS2001}
the results of Sections~\ref{sec-nr} and~\ref{sec-PF} are not new and have been
obtained earlier in \cite{FGS2001}.
However, in order to prove the higher order estimate \eqref{hoe-intro1} for the
no-pair operator we virtually have to re-derive it
for the semi-relativistic Pauli-Fierz operator by our own method anyway.
Moreover, we think that the arguments 
employed in Sections~\ref{sec-nr} and~\ref{sec-PF}
are more convenient and less involved than the procedure
carried through in \cite{FGS2001}.
The main text is followed by an appendix where we show that
the semi-relativistic Pauli-Fierz operator for a molecular
system with static nuclei is semi-bounded below, provided
that all Coulomb coupling constants are less than or 
equal to $2/\pi$. Moreover, we prove the same result
for a molecular no-pair operator assuming that all
Coulomb coupling constants are strictly less than
the critical coupling constant of the Brown-Ravenhall model \cite{EPS1996}.
The results of the appendix are based on corresponding estimates
for hydrogen-like atoms obtained in \cite{MatteStockmeyer2009a}.
(We remark that the considerably stronger stability of matter
of the second kind has been proven for a molecular
no-pair operator in \cite{LiebLoss2002} under more restrictive
assumptions on the involved physical parameters.)
No restrictions on the values of the fine-structure constant
or on the ultra-violet cut-off are imposed in the present article. 

The main new results of this paper
are Theorem~\ref{thm-hoe-allg} and its
corollaries which provide general criteria for 
the validity of higher order estimates
and Theorem~\ref{thm-hoe-np} where higher order estimates
for the no-pair operator are established.

\smallskip

\noindent
{\it Some frequently used notation.}
For $a,b\in\RR$, we write $a\wedge b:=\min\{a,b\}$
and $a\vee b:=\max\{a,b\}$. $\dom(T)$ denotes the domain
of some operator $T$ acting in some Hilbert space
and $\form(T)$ its form domain, when $T$ is semi-bounded below. 
$C(a,b,\ldots),C'(a,b,\ldots)$, etc. denote constants that 
depend only on the quantities $a,b,\ldots$ and whose value
might change from one estimate to another. 


\section{Higher order estimates: a general criterion}
\label{sec-general}

\noindent
The following theorem and its succeeding corollaries 
present the key idea behind of our method. They essentially reduce
the derivation of the higher order estimates to
the verification of a certain sequence of form bounds.
These form bounds can be verified easily
without any further induction argument
in the QED models treated in this paper. 

\begin{theorem}\label{thm-hoe-allg}
Let $H$ and $F_\ve$, $\ve>0$, be self-adjoint
operators in some Hilbert space $\sK$ 
such that $H\grg1$, $F_\ve\grg0$, and each $F_\ve$ is bounded. 
Let $m\in\NN\cup\{\infty\}$, let $\sD$ be a form
core for $H$, and assume that the following
conditions are fulfilled:
\begin{enumerate} 
\item[(a)] For every $\ve>0$,  
$F_\ve$ maps $\sD$ into $\form(H)$ and
there is some $c_\ve\in(0,\infty)$ such that
$$
\SPb{F_\ve\,\psi}{H\,F_\ve\,\psi}\,\klg\,c_\ve\,\SPn{\psi}{H\,\psi}\,,
\qquad \psi\in\sD\,.
$$
\item[(b)] There is some $c\in[1,\infty)$ such that, for all $\ve>0$,
$$
\SPn{\psi}{F_\ve^2\,\psi}\klg c^2\,\SPn{\psi}{H\,\psi}\,,
\qquad \psi\in\sD\,.
$$
\item[(c)] For every $n\in\NN$, $n<m$, there is some $c_n\in[1,\infty)$
such that, for all $\ve>0$,
\begin{align*}
\big|\SPn{&H\,\vp_1}{F_\ve^n\,\vp_2}
-\SPn{F_\ve^{n}\,\vp_1}{H\,\vp_2}\big|\nonumber
\\
&\klg c_n\,\big\{
\SPn{\vp_1}{H\,\vp_1}
+\SPn{F_\ve^{n-1}\,\vp_2}{H\,F_\ve^{n-1}\,\vp_2}
\big\}\,,\qquad\vp_1,\vp_2\in\sD.
\end{align*}
\end{enumerate}
Then it follows that, for every $n\in\NN$, $n<m+1$,
\begin{equation}\label{bea1allg}
\|\,F^n_\ve\,H^{-n/2}\,\|\,\klg\,C_n\,:=\,
4^{n-1}\,c^{n}\,\prod_{\ell=1}^{n-1} c_\ell\,.
\end{equation}
(An empty product equals $1$ by definition.)
\end{theorem}

\begin{proof}
We define
$$
T_\ve(n):=H^{1/2}\,[
F_\ve^{n-1}\,,\,H^{-1}]\,
H^{-(n-2)/2},\qquad n\in\{2,3,4,\ldots\}.
$$
$T_\ve(n)$ is well-defined and bounded because of 
the closed graph theorem and Condition (a),
which implies that $F_\ve\in\LO(\form(H))$, where 
$\form(H)=\dom(H^{1/2})$ is
equipped with the form norm.
We shall prove the following sequence
of assertions by induction on $n\in\NN$, $n<m+1$.
\begin{align}
A(n)\;:\Leftrightarrow\qquad&
\textrm{The bound \eqref{bea1allg} holds true and, if}\;n>3,\;
\textrm{we have}
\nonumber
\\ \label{bea1ballg}
&\forall\:\ve>0\::\quad\|T_\ve(n)\|\,\klg\,C_n/4c^2\,.
\end{align}
For $n=1$, 
the bound \eqref{bea1allg} is fulfilled with $C_1=c$ 
on account of Condition~(b).

Next, assume that $n\in\NN$, $n<m$, and that $A(1),\ldots,A(n)$ hold true.
To find a bound on $\|F^{n+1}_\ve\,H^{-(n+1)/2}\|$ we write
\begin{align}\label{lisa1}
&F_\ve^{n+1}\,H^{-(n+1)/2}\,=\,Q_1\,+\,Q_2
\end{align}
with
\begin{align*}
Q_1\,&:=\,
F_\ve\,H^{-1}\,F_\ve^{n}
\,H^{-(n-1)/2}\,,
\qquad
Q_2\,:=\,
F_\ve\,
\big[\,F_\ve^{n}\,,\,H^{-1}\,\big]\,
H^{-(n-1)/2}\,.
\end{align*}
By the induction hypothesis we have 
\begin{equation}\label{lisa2a}
\|Q_1\|\klg\|F_\ve\,H^{-1/2}\|\,\|H^{-1/2}\,F_\ve\,\|\,
\|F^{n-1}_\ve\,H^{-(n-1)/2}\|\klg c^{2}\,C_{n-1}
\,,
\end{equation}
where $C_0:=1$.
Moreover, we observe that
\begin{equation}\label{lisa2}
\|Q_2\|\,=\,\|F_\ve\,H^{-1/2}\,T_\ve(n+1)\|\,\klg\,
c\,\|T_\ve(n+1)\|\,.
\end{equation}
To find a bound on $\|T_\ve(n+1)\|$
we recall that $F_\ve$ maps the form domain of $H$
continuously into itself.
In particular, since $\sD$ is a form core
for $H$ the form bound appearing in Condition~(c) 
is available, for all $\vp_1,\vp_2\in\form(H)$.
Let $\phi,\psi\in\sD$.
Applying Condition~(c), extended in this way, with
\begin{align*}
\vp_1\,&=\,\delta^{1/2}\,H^{-1/2}\,\phi\in\form(H)\,,
\qquad
\vp_2\,=\,\delta^{-1/2}\,
H^{-(n+1)/2}\,\psi\in\form(H)\,,
\end{align*}
for some $\delta>0$, we obtain
\begin{align}
|\SPn{&\phi}{T_\ve(n+1)\,\psi}|\nonumber
\\
&=\,\nonumber
\big|\SPb{H\,H^{-1/2}\,\phi}{F_\ve^n\,H^{-(n+1)/2}\,\psi}
-\SPb{F_\ve^n\,H^{-1/2}\,\phi}{H\,H^{-(n+1)/2}\,\psi}\big|
\\
&\klg\,\nonumber
c_n\,\inf_{\delta>0}\big\{\,\delta\,\|\phi\|^2+\delta^{-1}\,
\|\{H^{1/2}\,F_\ve^{n-1}\,H^{-n/2}\}\,H^{-1/2}\,\psi\|^2\,\big\}
\\
&\klg\,\nonumber
2\,c_n\,\|\{H^{1/2}\,F_\ve^{n-1}\,H^{-n/2}\}\|\,\|\phi\|\,\|\psi\|
\,.
\end{align}
The operator $\{\cdots\}$ is just the identity when $n=1$.
For $n>1$, it can be written as
\begin{equation}\label{markus}
H^{1/2}\,F_\ve^{n-1}\,H^{-n/2}\,=\,
\{H^{-1/2}\,F_\ve\}\,F_\ve^{n-2}\,H^{-(n-2)/2}
\,+\,T_\ve(n)\,.
\end{equation}
Applying the induction hypothesis and $c,c_\ell\grg1$,
we thus get $\|T_\ve(2)\|\klg 2\,c_1$, $\|T_\ve(3)\|\klg 6\,c\,c_1c_2$, 
$\|T_\ve(4)\|\klg 14\,c^2c_1c_2c_3<C_4/4c^2$, and
\begin{align*}
c\,\|T_\ve(n+1)\|\,&=\,
c\,\sup\big\{\,|\SPn{\phi}{T_\ve(n+1)\,\psi}|\::\;\phi,\psi\in\sD\,,\;
\|\phi\|=\|\psi\|=1\,\big\}
\\
&\klg\,
2\,c_n\,(c^2\,C_{n-2}+C_n/4c)\,<\,c_n\,C_n\,=\,C_{n+1}/4c\,,\qquad n>3\,,
\end{align*}
since $c^2\,C_{n-2}\klg C_n/16$, for $n>3$.
Taking \eqref{lisa1}--\eqref{lisa2}
into account we arrive at
$\|F_\ve^2\,H^{-1}\|\klg c^2+2c\,c_1< C_2$, 
$\|F_\ve^3\,H^{-3/2}\|\klg c^3+6c^2c_1c_2<C_3$, and
$$
\|F^{n+1}_\ve\,H^{-(n+1)/2}\|\,<\,
c^2\,C_{n-2}+C_{n+1}/4c\,<\,C_{n+1}\,,\qquad n>3\,,
$$
which concludes the induction step.
\end{proof}

\begin{corollary}\label{cor-hoe-comm}
Assume that $H$ and $F_\ve$, $\ve>0$, are self-adjoint
operators in some Hilbert space $\mathscr{K}$ that fulfill
the assumptions of Theorem~\ref{thm-hoe-allg} with (c)
replaced by the stronger condition
\begin{enumerate}
\item[(c')] For every $n\in\NN$, $n<m$, there is some $c_n\in[1,\infty)$
such that, for all $\ve>0$,
\begin{align*}
\big|\SPn{&H\,\vp_1}{F_\ve^n\,\vp_2}
-\SPn{F_\ve^{n}\,\vp_1}{H\,\vp_2}\big|\nonumber
\\
&\klg c_n\,\big\{
\|\vp_1\|^2
+\SPn{F_\ve^{n-1}\,\vp_2}{H\,F_\ve^{n-1}\,\vp_2}
\big\}\,,\qquad\vp_1,\vp_2\in\sD.
\end{align*}
\end{enumerate}
Then, in addition to \eqref{bea1allg}, it follows that,
for $n\in\NN$, $n<m$,
$[F_\ve^n\,,\,H]\,H^{-n/2}$ defines a bounded sesquilinear
form with domain $\form(H)\times\form(H)$ and
\begin{equation}\label{hoe-comm}
\big\|\,[F_\ve^n\,,\,H]\,H^{-n/2}\,\big\|\,\klg\,C_n'\,:=\,
4^{n}\,c^{n-1}\,\prod_{\ell=1}^{n}c_\ell\,.
\end{equation} 
\end{corollary}

\begin{proof}
Again, the form bound in (c') is available, 
for all $\vp_1,\vp_2\in\form(H)$, whence
\begin{align*}
&\big|\SPn{H\,\phi}{F_\ve^n\,H^{-n/2}\,\psi}
-\SPn{F_\ve^{n}\,\phi}{H\,H^{-n/2}\,\psi}\big|
\\
&\klg\,
c_n\,\inf_{\delta>0}
\big\{\delta\,\|\phi\|^2+\delta^{-1}\,
\big\|H^{1/2}\,F_\ve^{n-1}\,H^{-n/2}\,\psi\big\|^2\big\}
\,\klg\,
2\,c_n\,\|H^{1/2}\,F_\ve^{n-1}\,H^{-n/2}\|\,,
\end{align*}
for all normalized $\phi,\psi\in\form(H)$.
The assertion now follows from \eqref{bea1allg},
\eqref{markus}, and the bounds on $\|T_\ve(n)\|$
given in the proof of Theorem~\ref{thm-hoe-allg}.
\end{proof}

\begin{corollary}\label{cor-hoe-allg}
Let $H\grg1$ and $A\grg0$ be two self-adjoint operators
in some Hilbert space $\sK$. Let $\kappa>0$, define
$$
f_\ve(t)\,:=\,t/(1+\ve\,t)\,,\quad t\grg0\,,
\qquad F_\ve:=f_\ve^\kappa(A)\,,
$$
for all $\ve>0$, and assume that $H$ and $F_\ve$, $\ve>0$,
fulfill the hypotheses of Theorem~\ref{thm-hoe-allg}, for some 
$m\in\NN\cup\{\infty\}$.
Then $\Ran(H^{-n/2})\subset\dom(A^{\kappa\,n})$, for every $n\in\NN$, $n<m+1$,
and
$$
\big\|\,A^{\kappa\,n}\,H^{-n/2}\,\big\|\,\klg\,4^{n-1}\,c^n
\prod_{\ell=1}^{n-1}c_\ell\,.
$$
If $H$ and $F_\ve$, $\ve>0$, fulfill the hypotheses of
Corollary~\ref{cor-hoe-comm}, then, for every $n\in\NN$, $n<m$,
it additionally follows that
$A^{\kappa\,n}\,H^{-n/2}$ maps $\dom(H)$ into itself so that 
$[A^{\kappa\,n}\,,\,H]\,H^{-n/2}$ is well-defined on $\dom(H)$, and
$$
\big\|\,[A^{\kappa\,n}\,,\,H]\,H^{-n/2}\,\big\|\,
\klg\,4^{n}\,c^{n-1}\,\prod_{\ell=1}^{n}c_\ell\,.
$$
\end{corollary}

\begin{proof}
Let $U:\sK\to L^2(\Omega,\mu)$ be a unitary transformation
such that $a=U\,A\,U^*$ is a maximal operator of
multiplication with some non-negative 
measurable function -- again called $a$ --
on some measure space $(\Omega,\fA,\mu)$.
We pick some $\psi\in\sK$, set $\phi_n:=U\,H^{-n/2}\,\psi$,
and apply the monotone convergence theorem to conclude that
\begin{align*}
\int_\Omega a(\omega)^{2\kappa\,n}\,
|\phi_n(\omega)|^2\,
d\mu(\omega)\,&=\,\lim_{\ve\searrow0}
\int_\Omega f_\ve^\kappa(a(\omega))^{2n}\,
|\phi_n(\omega)|^2\,d\mu(\omega)
\\
&=\,\lim_{\ve\searrow0}\|F_\ve^n\,H^{-n/2}\,\psi\|^2
\,\klg\,C_n\,\|\psi\|^2,
\end{align*}
for every $n\in\NN$, $n<m+1$, which implies the first assertion.
Now, assume that $H$ and $F_\ve$, $\ve>0$, fulfill
Condition~(c') of Corollary~\ref{cor-hoe-comm}.
Applying the dominated convergence theorem in the spectral
representation introduced above we see that
$F_\ve^n\,\psi\to A^{\kappa\,n}\,\psi$, for every
$\psi\in\dom(A^{\kappa\,n})$. Hence, \eqref{hoe-comm} and
$\Ran(H^{-n/2})\subset\dom(A^{\kappa\,n})$ imply,
for $n<m$ and $\phi,\psi\in\dom(H)$,
\begin{align*}
\big|\SPb{&\phi}{A^{\kappa\,n}H^{-n/2}\,H\,\psi}-
\SPb{H\,\phi}{A^{\kappa\,n}\,H^{-n/2}\,\psi}\big|
\\
&=\,
\lim_{\ve\searrow0}
\big|\SPb{F_\ve^n\,\phi}{H\,H^{-n/2}\,\psi}-
\SPb{H\,\phi}{F_\ve^n\,H^{-n/2}\,\psi}\big|
\\
&\klg\,
\limsup_{\ve\searrow0}\big\|\,[F_\ve^n\,,\,H]\,H^{-n/2}\,\big\|\,
\|\phi\|\,\|\psi\|\,\klg\,
C_n'\,\|\phi\|\,\|\psi\|\,.
\end{align*}
Thus,
$|\SPn{H\,\phi}{A^{\kappa\,n}\,H^{-n/2}\,\psi}|\klg
\|\phi\|\,\|A^{\kappa\,n}\,H^{-n/2}\|\,\|H\,\psi\|+C_n'\|\phi\|\,\|\psi\|$,
for all $\phi,\psi\in\dom(H)$. In particular,
$A^{\kappa\,n}\,H^{-n/2}\,\psi\in\dom(H^*)=\dom(H)$, for
all $\psi\in\dom(H)$, and the second asserted bound holds true. 
\end{proof}


\section{Commutator estimates}
\label{sec-comm}

\noindent
In this section we derive operator norm bounds on
commutators involving the quantized vector potential, $\V{A}$,
the radiation field energy, $\Hf$, and the Dirac operator, $\DA$.
The underlying Hilbert space is
$$
\HR\,:=\,
L^2(\RR^3_\V{x}\times\ZZ_4)\otimes\Fock
\,=\,\int^\oplus_{\RR^3}\CC^4\otimes\Fock\,d^3\V{x}\,,
$$
where the bosonic Fock space, $\Fock$,
is modeled over the one-photon Hilbert space
$$
\Fock^{(1)}\,:=\,L^2(\cA\times\ZZ_2,dk)\,,\qquad 
\int dk\,:=\,\sum_{\lambda\in\ZZ_2}\int_\cA d^3\V{k}\,.
$$
With regards to the applications in \cite{KMS2009b} we define 
$\cA:=\{\V{k}\in\RR^3:\,|\V{k}|\grg m\}$, for some $m\grg0$.
We thus have
$$
\Fock\,=\,\bigoplus_{n=0}^\infty\Fock^{(n)},\qquad \Fock^{(0)}:=\CC\,,
\quad\Fock^{(n)}:=
\cS_n\,L^2\big((\cA\times\ZZ_2)^n\big),\;n\in\NN,
$$
where $\cS_n=\cS_n^2=\cS_n^*$ is given by 
$$
(\cS_n\,\psi^{(n)})(k_1,\ldots,k_n)\,:=\,
\frac{1}{n!}\sum_{\pi\in\mathfrak{S}_n}\psi^{(n)}(k_{\pi(1)},\ldots,k_{\pi(n)}),
\quad \psi^{(n)}\in L^2\big((\cA\times\ZZ_2)^n\big),
$$
$\mathfrak{S}_n$ denoting the group of permutations of $\{1,\ldots,n\}$. 
The vector potential is determined by a certain 
vector-valued function, $\V{G}$,
called the form factor.

\begin{hypothesis}\label{hyp-G}
The dispersion relation, 
$\omega:\cA\to[0,\infty)$, is a measurable function 
such that $0<\omega(k):=\omega(\V{k})\klg|\V{k}|$, for 
$k=(\V{k},\lambda)\in\cA\times\ZZ_2$ with $\V{k}\not=0$.
For every 
$k\in(\cA\setminus\{0\})\times\ZZ_2$ and $j\in\{1,2,3\}$, 
$G^{(j)}(k)$
is a bounded continuously differentiable function,
$\RR^3_{\V{x}}\ni\V{x}\mapsto G^{(j)}_{\V{x}}(k)$,
such that the map $(\V{x},k)\mapsto G_{\V{x}}^{(j)}(k)$
is measurable and
$G_\V{x}^{(j)}(-\V{k},\lambda)=\ol{G_\V{x}^{(j)}(\V{k},\lambda)}$,
for almost every $\V{k}$ and all $\V{x}\in\RR^3$ and $\lambda\in\ZZ_2$.
Finally, there exist $d_{-1},d_0,d_1,\ldots\in(0,\infty)$ such that
\begin{align}
\label{def-d3}
2\int\omega(k)^{\ell}\,\|\V{G}(k)\|^2_\infty\,dk
\,&\klg\,d_\ell^2\,,\qquad \ell\in\{-1,0,1,2,\ldots\}\,,
\\
\label{hyp-rotG}
2\int\omega(k)^{-1}\,\|\nabla_{\V{x}}\wedge\V{G}(k)\|^2_\infty\,dk
\,&\klg\,d_1^2\,,
\end{align}
where
$\V{G}=(G^{(1)},G^{(2)},G^{(3)})$ and
$\|\V{G}(k)\|_\infty:=\sup_{\V{x}}|\V{G}_{\V{x}}(k)|$, etc.
\end{hypothesis}

\begin{example}
In the physical applications the form factor is often given as
\begin{equation}\label{Gphys}
\V{G}^{e,\UV}_\V{x}(k)\,
:=\,-e\,\frac{\id_{\{|\V{k}|\klg\UV\}}}{2\pi\sqrt{|\V{k}|}}
\,e^{-i\V{k}\cdot\V{x}}\,\veps(k),
\end{equation}
for $(\V{x},k)\in\RR^3\times(\RR^3\times\ZZ_2)$
with $\V{k}\not=0$.
Here the physical units are chosen such that 
energies are measured in units of the rest energy of the electron.
Length are measured in units of one
Compton wave length divided by $2\pi$.
The parameter
$\UV>0$ is an ultraviolet cut-off and the square of the
elementary charge, $e>0$, equals Sommerfeld's
fine-structure constant in these units; we have $e^2\approx1/137$
in nature.
The polarization vectors, $\veps(\V{k},\lambda)$, $\lambda\in\ZZ_2$,
are homogeneous of degree zero in $\V{k}$ such that
$\{\mr{\V{k}},\veps(\mr{\V{k}},0),\veps(\mr{\V{k}},1)\}$
is an orthonormal basis of $\RR^3$,
for every $\mr{\V{k}}\in S^2$.
This corresponds to the Coulomb gauge for
$\nabla_\V{x}\cdot\V{G}^{e,\UV}=0$. 
We remark that the vector fields 
$S^2\ni\mr{\V{k}}\mapsto\veps(\mr{\V{k}},\lambda)$
are necessarily discontinuous.\hfill$\diamond$
\end{example}

\smallskip

\noindent
It is useful to work with more general form factors
fulfilling Hypothesis~\ref{hyp-G}
since in the study of the existence of ground states
in QED one usually encounters truncated and discretized
versions of the physical choice
$\V{G}^{e,\UV}$. For the applications in \cite{KMS2009b} 
it is necessary to know that
the higher order estimates
established here hold true uniformly in the
involved parameters and Hypothesis~\ref{hyp-G} is convenient
way to handle this.

We recall the definition of the creation
and the annihilation operators
of a photon state $f\in\Fock^{(1)}$,
\begin{align*}
(\ad(f)\,\psi)^{(n)}(k_1,\ldots,k_n)\,&=\,
n^{-1/2}\sum_{j=1}^nf(k_j)\,\psi^{(n-1)}(\ldots,k_{j-1},k_{j+1},\ldots)\,,
\quad n\in\NN\,,
\\
(a(f)\,\psi)^{(n)}(k_1,\ldots,k_n)\,&=\,
(n+1)^{1/2}\int\ol{f}(k)\,\psi^{(n+1)}(k,k_1,\ldots,k_n)\,dk\,,
\;\;\;\, n\in\NN_0\,,
\end{align*}
and $(\ad(f)\,\psi)^{(0)}=0$, $a(f)\,(\psi^{(0)},0,0,\ldots)=0$, for all
$\psi=(\psi^{(n)})_{n=0}^\infty\in\Fock$ such that the right hand
sides again define elements of $\Fock$. 
$\ad(f)$ and $a(f)$ are formal adjoints of each other on the
dense domain
$$
\sC_0\,:=\,\CC\oplus\bigoplus_{n=1}^\infty \cS_n\,
L^\infty_{\mathrm{comp}}\big((\cA\times\ZZ_2)^n\big)
\,.
\qquad(\textrm{Algebraic direct sum.})
$$
For a three-vector of functions 
$\V{f}=(f^{(1)},f^{(2)},f^{(3)})\in(\Fock^{(1)})^3$, we write
$a^\sharp(\V{f}):=(a^\sharp(f^{(1)}),a^\sharp(f^{(2)}),a^\sharp(f^{(3)}))$,
where $a^\sharp$ is $\ad$ or $a$.
Then the quantized vector potential is the triplet of operators
given by
$$
\V{A}\,\equiv\,\V{A}(\V{G})\,:=\,\ad(\V{G})+a(\V{G})\,,\qquad
a^\sharp(\V{G})\,:=\,\int^\oplus_{\RR^3}
\id_{\CC^4}\otimes a^\sharp(\V{G}_{\V{x}})\,d^3\V{x}\,.
$$
The radiation field energy is the direct sum
$\Hf=\bigoplus_{n=0}^\infty d\Gamma^{(n)}(\omega):\dom(\Hf)\subset\Fock\to\Fock$,
where $d\Gamma^{(0)}(\omega):=0$, and $d\Gamma^{(n)}(\omega)$ denotes
the maximal multiplication operator in $\Fock^{(n)}$
associated with the symmetric function 
$(k_1,\ldots,k_n)\mapsto\omega(k_1)+\dots+\omega(k_n)$.
By the permutation symmetry and Fubini's theorem we thus have
\begin{equation}\label{Hf=dGamma}
\SPb{\Hf^{1/2}\,\phi}{\Hf^{1/2}\,\psi}\,=\,\int\omega(k)\,
\SPn{a(k)\,\phi}{a(k)\,\psi}\,dk\,,
\qquad \phi,\psi\in\dom(\Hf^{1/2})\,,
\end{equation}
where we use the notation
$$
(a(k)\,\psi)^{(n)}(k_1,\dots,k_n)\,=\,
(n+1)^{1/2}\,\psi^{(n+1)}(k,k_1,\dots,k_n)\,,\quad n\in\NN_0\,,
$$
almost everywhere,
and $a(k)\,(\psi^{(0)},0,0,\ldots)=0$.
For a measurable function $f:\RR\to\RR$ and 
$\psi\in\dom(f(\Hf))$, the following identity in $\Fock^{(n)}$,
$$
(a(k)\,f(\Hf)\,\psi)^{(n)}\,=\,
f\big(\omega(k)+d\Gamma^{(n)}(\omega)\big)\,(a(k)\,\psi)^{(n)}\,,
\qquad n\in\NN_0\,,
$$
valid for almost every $k$,
is called the pull-through formula.
Finally, we let $\alpha_1,\alpha_2,\alpha_3$, and $\beta:=\alpha_0$
denote hermitian four times four matrices that fulfill the
Clifford algebra relations
\begin{equation}\label{Clifford}
\alpha_i\,\alpha_j\,+\,\alpha_j\,\alpha_i\,=\,2\,\delta_{ij}\,\id\,,
\qquad i,j\in\{0,1,2,3\}\,.
\end{equation}
They act on the second tensor factor in 
$L^2(\RR^3_\V{x}\times\ZZ_4)=L^2(\RR^3_\V{x})\otimes\CC^4$.
As a consequence of \eqref{Clifford} and the $C^*$-equality we have
\begin{equation}\label{C*}
\|\valpha\cdot \V{v}\|_{\LO(\CC^4)}=|\V{v}|\,,\quad \V{v}\in\RR^3\,,\qquad
\|\valpha\cdot \V{z}\|_{\LO(\CC^4)}\klg\sqrt{2}\,|\V{z}|\,,\quad 
\V{z}\in\CC^3\,,
\end{equation}
where $\valpha\cdot\V{z}:=\alpha_1\,z^{(1)}+\alpha_2\,z^{(2)}+\alpha_3\,z^{(3)}$,
for $\V{z}=(z^{(1)},z^{(2)},z^{(3)})\in\CC^3$.
A standard exercise using the inequality in \eqref{C*}, the
Cauchy-Schwarz inequality, and the canonical commutation relations,
$$
[a^\sharp(f)\,,\,a^\sharp(g)]=0\,,\qquad[a(f)\,,\,\ad(g)]=\SPn{f}{g}\,\id\,,
\qquad f,g\in\Fock^{(1)}\,,
$$
reveals that every $\psi\in\dom(\Hf^{1/2})$ belongs
to the domain of $\valpha\cdot a^\sharp(\V{G})$ and 
\begin{equation}\label{rb-aHf}
\|\valpha\cdot a(\V{G})\,\psi\|\klg d_{-1}\,\|\Hf^{1/2}\,\psi\|\,,
\quad
\|\valpha\cdot \ad(\V{G})\,\psi\|^2\klg
d_{-1}^2\,\|\Hf^{1/2}\,\psi\|^2+d_0^2\,\|\psi\|^2.
\end{equation}
(Here and in the following we identify $\Hf\equiv\id\otimes\Hf$, etc.)
These relative bounds imply that
$\valpha\cdot\V{A}$ is symmetric on the domain $\dom(\Hf^{1/2})$.

The operators whose norms are estimated in 
\eqref{xaver0} and the following lemmata
are always well-defined \`a priori on
the following dense subspace of $\HR$,
$$
\core\,:=\,C_0^\infty(\RR^3\times\ZZ_4)\otimes\sC_0\,.
\qquad (\textrm{Algebraic tensor product.})
$$
Given some $E\grg1$ we set
\begin{equation}\label{det-HT}
\HT\,:=\,\Hf+E
\end{equation}
in the sequel.
We already know from \cite{MatteStockmeyer2009a} that, for
every $\nu\grg0$,
there is some constant, $C_\nu\in(0,\infty)$, such that
\begin{equation}\label{xaver0}
\big\|\,[\valpha\cdot\V{A}\,,\,\HT^{-\nu}]\,\HT^\nu\,\big\|\,
\klg\,C_\nu/E^{1/2}\,,\qquad E\grg1\,.
\end{equation}
In our first lemma we derive a generalization of \eqref{xaver0}.
Its proof resembles the one of \eqref{xaver0} given
in \cite{MatteStockmeyer2009a}.
Since we shall encounter many similar but slightly
different commutators in the applications
it makes sense to introduce the numerous parameters 
that obscure its statement (but simplify its proof).

\begin{lemma}\label{le-clara1}
Assume that $\omega$ and $\V{G}$ fulfill Hypothesis~\ref{hyp-G}.
Let $\ve\grg0$, $E\grg1$, 
$\kappa,\nu\in\RR$, $\gamma,\delta,\sigma,\tau\grg0$, 
such that $\gamma+\delta+\sigma+\tau\klg1/2$, and define 
\begin{equation}\label{clara1}
f_\ve(t)\,:=\,\frac{t+E}{1+\ve t+\ve E}\,,
\qquad t\in[0,\infty)\,.
\end{equation}
Then the operator
$
\HT^{\nu+\gamma}\,f_\ve^\sigma(\Hf)\,
[\valpha\cdot\V{A}\,,\,f_\ve^{\kappa}(\Hf)]
\,\HT^{-\nu+\delta}f_\ve^{-\kappa+\tau}(\Hf)
$, defined \`a priori on $\core$,
extends to a bounded operator on $\HR$ and 
\begin{align}
\big\|\,\HT^{\nu+\gamma}\,f_\ve^\sigma(\Hf)
\,[\valpha\cdot\V{A}\,,\,f_\ve^{\kappa}(\Hf)]\,
\HT^{-\nu+\delta}&\,f_\ve^{-\kappa+\tau}(\Hf)
\,\big\|\,\nonumber
\\
&\klg\,|\kappa|\,2^{(\rho+1)/2}\,(d_1+d_{\rho})\,E^{\gamma+\delta+\sigma+\tau-1/2}
\label{clara1a}\,,
\end{align}
where $\rho$ is
the smallest integer greater or equal to $3+2|\kappa|+2|\nu|$.
\end{lemma}

\begin{proof}
We notice that all operators $\HT^s$ and $f_\ve^s(\Hf)$
leave the dense subspace $\core$ invariant and 
that $\valpha\cdot a^\sharp(\V{G})$ maps $\sD$ into
$\dom(\HT^s)$, for every $s\in\RR$. Now, let $\vp,\psi\in\core$.
Then 
\begin{align}
\SPb{\vp}{&\HT^{\nu+\gamma}\,f_\ve^\sigma(\Hf)\,[\valpha\cdot\V{A}\,,\,
f_\ve^{\kappa}(\Hf)]\,
\HT^{-\nu+\delta}\,f_\ve^{-\kappa+\tau}(\Hf)\,\psi}\nonumber
\\
&=\,\label{clara-99}
\SPb{\vp}{
\HT^{\nu+\gamma}\,f_\ve^\sigma(\Hf)\,[\valpha\cdot a(\V{G})\,,\,
f_\ve^{\kappa}(\Hf)]\,\HT^{-\nu+\delta}
\,f_\ve^{-\kappa+\tau}(\Hf)\,\psi}
\\
&\quad- \label{clara2}
\SPb{f_\ve^{-\kappa+\tau}(\Hf)\,\HT^{-\nu+\delta}\,
[\valpha\cdot a(\V{G})\,,\,f_\ve^{\kappa}(\Hf)]\,f_\ve^\sigma(\Hf)\,
\HT^{\nu+\gamma}\,\vp}{\psi}\,.
\end{align}
For almost every $k$,
the pull-through formula yields the following representation,
\begin{align*}
\HT^{\nu+\gamma}\,f_\ve^\sigma(\Hf)\,
[a(k)\,,\,f_\ve^{\kappa}(\Hf)]
\,\HT^{-\nu+\delta}\,f_\ve^{-\kappa+\tau}(\Hf)\,\psi
\,=\,
F(k;\Hf)\,a(k)\,\HT^{-1/2}\,\psi\,,
\end{align*}
where
\begin{align*}
F(k;t)\,&:=\,
(t+E)^{\nu+\gamma}\,f_\ve^\sigma(t)\,\big(f_\ve^{\kappa}(t+\omega(k))-
f_\ve^{\kappa}(t)\big)
\\
&\qquad\cdot
(t+E+\omega(k))^{-\nu+\delta+1/2}\,f_\ve^{-\kappa+\tau}(t+\omega(k))
\\
&=\,\Big(\frac{t+E}{t+E+\omega(k)}\Big)^\nu(t+E)^\gamma\,
(t+E+\omega(k))^{\delta+1/2}
\\
&\qquad\cdot\int_0^1\frac{d}{ds}f_\ve^{\kappa}(t+s\,\omega(k))\,ds
\:\frac{f_\ve^\sigma(t)\,f_\ve^\tau(t+\omega(k))}{f_\ve^{\kappa}(t+\omega(k))}\,,
\end{align*} 
for $t\grg0$. We compute
\begin{equation}\label{abl-fe}
\frac{d}{ds}f_\ve^{\kappa}(t+s\,\omega(k))\,=\,
\frac{\kappa\,\omega(k)\,
f_\ve^{\kappa}(t+s\,\omega(k))}{(t+s\,\omega(k)+E)
(1+\ve\, t+\ve\,s\,\omega(k)+\ve\,E)}\,.
\end{equation}
Using that $f_\ve$ is increasing in $t\grg0$ and that
$$
(t+\omega(k)+E)/(t+s\,\omega(k)+E)\,\klg\,1+\omega(k)\,,
\qquad s\in[0,1]\,,
$$
thus
$$
f_\ve^\kappa(t+s\,\omega(k))/f_\ve^\kappa(t+\omega(k))\,\klg\,
(1+\omega(k))^{-(0\wedge\kappa)}\,,\qquad s\in[0,1]\,,
$$
it is elementary to verify that
\begin{align*}
|F_\ve(k;t)|\,&\klg\,|\kappa|\,\omega(k)\,
(1+\omega(k))^{\delta+\tau-(0\wedge\kappa)-(0\wedge\nu)+1/2}
\,E^{\gamma+\delta+\sigma+\tau-1/2}\,,
\end{align*} 
for all $t\grg0$ and $k$.
We deduce that
the term in \eqref{clara-99} can be estimated as
\begin{align}
&\big|\SPb{\vp}{\nonumber
\HT^{\nu+\gamma}\,f_\ve^\sigma(\Hf)\,
[\valpha\cdot a(\V{G})\,,\,f_\ve^{\kappa}(\Hf)]\,
\HT^{-\nu+\delta}\,f_\ve^{-\kappa+\tau}(\Hf)\,\psi}\big|
\\
&\klg\nonumber
\int\|\vp\|\,\big\|\valpha\cdot\V{G}(k)\,\HT^{\nu+\gamma}\,f_\ve^\sigma(\Hf)\,
[a(k)\,,\,f_\ve^{\kappa}(\Hf)]\,\HT^{-\nu+\delta}
\,f_\ve^{-\kappa+\tau}(\Hf)
\,\psi\big\|\,dk
\\
&\klg\nonumber
\sqrt{2}\int\|\vp\|\,\|\V{G}(k)\|_\infty\,\|F_\ve(k;\Hf)\|\,
\|a(k)\,\HT^{-1/2}\,\psi\|\,dk
\\
&\klg\nonumber
|\kappa|\,\sqrt{2}\Big(\int\omega(k)\,
(1+\omega(k))^{2(\delta+\tau)-(0\wedge2\kappa)-(0\wedge2\nu)+1}\,
\|\V{G}(k)\|^2_\infty\,dk\Big)^{1/2}
\\
& \qquad \cdot\nonumber
\Big(\int\omega(k)\,\big\|\,a(k)\,\HT^{-1/2}\,\psi\,\big\|^2\,dk\Big)^{1/2}
\,\|\vp\|\,E^{\gamma+\delta+\sigma+\tau-1/2}
\\\label{clara77}
&\klg\,|\kappa|\,2^{(\rho-1)/2}\,(d_{1}+d_{\rho})\,\|\vp\|\,
\big\|\,\Hf^{1/2}\,\HT^{-1/2}\,\psi\,\big\|\,E^{\gamma+\delta+\sigma+\tau-1/2}\,.
\end{align}
In the last step we used $\delta+\tau\klg1/2$ and applied \eqref{Hf=dGamma}.
\eqref{clara77} immediately gives a bound on
the term in \eqref{clara2}, too.
For we have
\begin{align*}
f_\ve^{-\kappa+\tau}(\Hf)&\,\HT^{-\nu+\delta}\,
[\valpha\cdot a(\V{G})\,,\,f_\ve^{\kappa}(\Hf)]\,f_\ve^\sigma(\Hf)\,
\HT^{\nu+\gamma}\,\vp
\\
&=\,
\HT^{-\nu+\delta}\,f_\ve^\tau(\Hf)\,
[f_\ve^{-\kappa}(\Hf)\,,\,\valpha\cdot a(\V{G})]\,\HT^{\nu+\gamma}\,
f_\ve^{\kappa+\sigma}(\Hf)\,\vp\,,
\end{align*} 
which after the replacements
$(\nu,\kappa,\gamma,\delta,\sigma,\tau)\mapsto
(-\nu,-\kappa,\delta,\gamma,\tau,\sigma)$ and $\vp\mapsto-\psi$
is precisely the term we just have treated.
\end{proof}

\smallskip

\noindent
Lemma~\ref{le-clara1} provides all the
information needed to apply Corollary~\ref{cor-hoe-allg}
to non-relativistic QED. For the application of Corollary~\ref{cor-hoe-allg}
to the non-local semi-relativistic models of QED it is necessary
to study commutators that involve resolvents and sign functions of the
Dirac operator,
$$
\DA\,:=\,\valpha\cdot(-i\nabla_\V{x}+\V{A})+\beta\,.
$$
An application of Nelson's commutator theorem with test
operator $-\Delta+H_f+1$ 
shows that $\DA$ is essentially self-adjoint on $\core$.
The spectrum of its unique closed extension, again denoted
by the same symbol, is contained in the union
of two half-lines, $\spec[\DA]\subset(-\infty,-1]\cup[1,\infty)$.
In particular, it makes sense to define
$$
\RA{iy}\,:=\,(\DA-iy)^{-1}\,,\qquad y\in\RR\,,
$$
and the spectral calculus yields
$$
\|\RA{iy}\|\klg(1+y^2)^{-1/2},\qquad
\int_\RR\big\|\,|\DA|^{1/2}\,\RA{iy}\,\psi\,\big\|^2\,\frac{dy}{\pi}\,=\,
\|\psi\|^2,\;\;\psi\in\HR.
$$
The next lemma is a straightforward extension of 
\cite[Corollary~3.1]{MatteStockmeyer2009a} where it is also
shown that $\RA{iy}$ maps $\dom(\Hf^\nu)$ into itself,
for every $\nu>0$.

\begin{lemma}\label{le-clara1b}
Assume that $\omega$ and $\V{G}$ fulfill Hypothesis~\ref{hyp-G}.
Then, for all $\kappa,\nu\in\RR$, 
we find $k_i\equiv k_i(\kappa,\nu,d_1,d_\rho)\in[1,\infty)$, $i=1,2$,
such that, for all $y\in\RR$, $\ve\grg0$, and $E\grg k_1$,
there exist 
$\Upsilon_{\kappa,\nu}(iy),\wt{\Upsilon}_{\kappa,\nu}(iy)\in\LO(\HR)$
satisfying
\begin{align}
\RA{iy}\,\HT^{-\nu}\,f^{-\kappa}_\ve(\Hf)\,
&=\,\label{eva2001}
\HT^{-\nu}\,f^{-\kappa}_\ve(\Hf)\,\RA{iy}\,\Upsilon_{\kappa,\nu}(iy)
\\
&=\,\label{eva2001b}
\HT^{-\nu}\,f^{-\kappa}_\ve(\Hf)\,\wt{\Upsilon}_{\kappa,\nu}(iy)\,\RA{iy}
\,,
\end{align}
on $\dom(\HT^{-\nu})$,
and 
$\|\Upsilon_{\kappa,\nu}(iy)\|,
\|\wt{\Upsilon}_{\kappa,\nu}(iy)\|\klg k_2$,
where $\rho$ is defined in Lemma~\ref{le-clara1}.
\end{lemma}

\begin{proof}
Without loss of generality we may assume that $\ve>0$
for otherwise we could simply replace $\nu$ by $\nu+\kappa$
and $f_0^\kappa$ by $f_0^0=1$. 
First, we assume in addition that $\nu\grg0$.
We observe that
\begin{align}\nonumber
T_0\,:=\,\big[\,&\HT^{-\nu}\,f_\ve^{-\kappa}(\Hf)\,,\,\valpha\cdot\V{A}\,\big]
\,\HT^\nu\,f_\ve^\kappa(\Hf)\,=\,T_1\,+\,T_2
\end{align}
on $\sD$, where
\begin{align*}
T_1\,&:=\,
[\HT^{-\nu}\,,\,\valpha\cdot\V{A}]\,\HT^\nu
\,,\qquad T_2\,:=\,
\HT^{-\nu}\,[f_\ve^{-\kappa}(\Hf)\,,\,\valpha\cdot\V{A}]\,f_\ve^\kappa(\Hf)
\,\HT^\nu.
\end{align*}
Due to \eqref{xaver0} (or \eqref{clara1a} with $\ve=0$)
the operator $T_1$ extends to a bounded
operator on $\HR$ and $\|T_1\|\klg C_\nu/E^{1/2}$.
According to \eqref{clara1a} we further have
$\|T_2\|\klg C_{\kappa,\nu}(d_1+d_\rho)/E^{1/2}$. 
We pick some $\phi\in\core$ and compute
\begin{align}
\big[\,&\RA{iy}\,,\,\HT^{-\nu}\,f^{-\kappa}_\ve(\Hf)\,\big]
\,(\DA-iy)\,\phi\nonumber
\,=\,\RA{iy}\,\big[\,\HT^{-\nu}\,f^{-\kappa}_\ve(\Hf)\,,\,\DA\,\big]\,\phi
\\  \nonumber
&=\,\RA{iy}\,T_0\,\HT^{-\nu}\,f^{-\kappa}_\ve(\Hf)\,\phi
\\  \label{fred1}
&=\,
\RA{iy}\,\ol{T}_0\,\HT^{-\nu}\,f^{-\kappa}_\ve(\Hf)
\,\RA{iy}\,(\DA-iy)\,\phi\,.
\end{align}
Since $(\DA-iy)\,\core$ is dense in $\HR$ and since $\HT^{-\nu}$
and $f_\ve^\kappa(\Hf)$ are bounded (here we use that $\nu\grg0$
and $\ve>0$),
this identity implies
$$
\RA{iy}\,\HT^{-\nu}\,f^{-\kappa}_\ve(\Hf)\,=\,\big(\id+\RA{iy}\,\ol{T}_0\big)\,
\HT^{-\nu}\,f^{-\kappa}_\ve(\Hf)\,\RA{iy}\,.
$$
Taking the adjoint of the previous identity and replacing
$y$ by $-y$ we obtain
\begin{equation*}
\HT^{-\nu}\,f^{-\kappa}_\ve(\Hf)\,\RA{iy}\,=\,
\RA{iy}\,\HT^{-\nu}\,f^{-\kappa}_\ve(\Hf)\,
(\id+T_0^*\,\RA{iy})\,.
\end{equation*}
In view of the norm bounds on $T_1$ and $T_2$ we see that
\eqref{eva2001} and \eqref{eva2001b} are valid with 
${\Upsilon}_{\kappa,\nu}(iy):=\sum_{\ell=0}^\infty\{-T_0^*\,\RA{iy}\}^\ell$
and $\wt{\Upsilon}_{\kappa,\nu}(iy):=\sum_{\ell=0}^\infty\{-\RA{iy}\,T_0^*\}^\ell$,
provided that $E$ is sufficiently large, depending only
on $\kappa,\nu,d_1$, and $d_\rho$, such that the
Neumann series converge.

Now, let $\nu<0$. Then we write $T_0$ on the domain $\sD$ as
$$
T_0\,=\,\HT^{-\nu}\,f_\ve^{-\kappa}(\Hf)\,
\big[\valpha\cdot\V{A}\,,\,\HT^\nu\,f_\ve^\kappa(\Hf)\big]\,,
$$
and deduce that
$$
\RA{iy}\,\HT^{\nu}\,f_\ve^{\kappa}(\Hf)\,(\id+\ol{T}_0\,\RA{iy})\,=\,
\HT^\nu\,f_\ve^\kappa(\Hf)\,\RA{iy}
$$
by a computation analogous to \eqref{fred1}.
Taking the adjoint of this identity with $y$
replaced by $-y$ we get
\begin{equation*}
\big(\id+\RA{iy}\,T_0^*\big)\,\HT^{\nu}\,f_\ve^{\kappa}\,\RA{iy}\,=\,
\RA{iy}\,\HT^\nu\,f_\ve^\kappa(\Hf)\,.
\end{equation*}
Next, we invert $\id+\RA{iy}\,T_0^*$ by means of the same
Neumann series as above. As a result we obtain
$$
\HT^\nu f_\ve^\kappa(\Hf)\,\RA{iy}=
\RA{iy}\,\Upsilon_{\kappa,\nu}(iy)\,\HT^\nu f_\ve^\kappa(\Hf)
=
\wt{\Upsilon}_{\kappa,\nu}(iy)\,\RA{iy}\,\HT^{\nu}f_\ve^\kappa(\Hf),
$$
where the definition of $\Upsilon_{\kappa,\nu}$
and $\wt{\Upsilon}_{\kappa,\nu}$ has been extended
to negative $\nu$.
It follows that 
$\RA{iy}\,\Upsilon_{\kappa,\nu}(iy)=\wt{\Upsilon}_{\kappa,\nu}(iy)\,\RA{iy}$
maps $\dom(\HT^{-\nu})=\dom(\HT^{-\nu}\,f_\ve^{-\kappa}(\Hf))
=\Ran(\HT^\nu\, f_\ve^\kappa(\Hf))$
into itself and that \eqref{eva2001} and \eqref{eva2001b}
still hold true when $\nu$ is negative.
\end{proof}

\smallskip

\noindent
In order to control the Coulomb singularity $1/|\V{x}|$
in terms of $|\DA|$ and $\Hf$ in the proof of the following
corollary, we shall employ
the bound \cite[Theorem~2.3]{MatteStockmeyer2009a}
\begin{equation}\label{maria}
\frac{2}{\pi}\frac{1}{|\V{x}|}\,\klg\,
|\DA|+\Hf+k\,d_1^2\,,
\end{equation}
which holds true in sense of quadratic forms
on $\form(|\DA|)\cap\form(\Hf)$.
Here $k\in(0,\infty)$ is some universal constant.
We abbreviate the sign function of the Dirac operator,
which can be represented as a strongly convergent
principal value \cite[Lemma~VI.5.6]{Kato}, by
\begin{equation}\label{for-sgn}
\SA\,\psi\,:=\,\DA\,|\DA|^{-1}\,\psi
\,=\,\lim_{\tau\to\infty}\int_{-\tau}^\tau \RA{iy}\,\psi\,\frac{dy}{\pi}\,.
\end{equation}
We recall from \cite[Lemma~3.3]{MatteStockmeyer2009a}
that $\SA$ maps $\dom(\Hf^\nu)$ into itself, for
every $\nu>0$. This can also be read off from the
proof of the next corollary.

\begin{corollary}\label{cor-clara1}
Assume that $\omega$ and $\V{G}$ fulfill Hypothesis~\ref{hyp-G}.
Let $\kappa,\nu\in\RR$.
Then we find some $C\equiv C(\kappa,\nu,d_1,d_\rho)\in(0,\infty)$
such that, for all $\gamma,\delta,\sigma,\tau\grg0$
with $\gamma+\delta+\sigma+\tau\klg1/2$ and all $\ve\grg0$,
$E\grg k_1$,
\begin{align}
\big\|\,\HT^{\nu}\,f_\ve^\kappa(\Hf)\,\SA\,
\HT^{-\nu}\,f_\ve^{-\kappa}(\Hf)\,\big\|
\,&\klg\,C,\label{clara1bb}
\\
\big\|\,|\DA|^{1/2}\,
\HT^{\nu+\gamma}\,f_\ve^\sigma(\Hf)\,[\SA\,,\,f_\ve^\kappa(\Hf)]
\,\HT^{-\nu+\delta}\,f_\ve^{-\kappa+\tau}(\Hf)\,\big\|\,
&\klg\,C\label{clara1b},
\\
\big\|\,|\V{x}|^{-1/2}\,
\HT^{\nu}\,f_\ve^\sigma(\Hf)\,[\SA\,,\,f_\ve^\kappa(\Hf)]\,\HT^{-\nu-\sigma-\tau}
\,f_\ve^{-\kappa+\tau}(\Hf)\,\big\|\,
&\klg\,C\label{clara1c}.
\end{align}
($k_1$ is the constant appearing in Lemma~\ref{le-clara1b},
$\HT$ is given by \eqref{det-HT}, $f_\ve$ by \eqref{clara1}.)
\end{corollary}

\begin{proof}
First, we prove \eqref{clara1b}.
Using
\eqref{for-sgn}, writing
\begin{equation*}
[\RA{iy}\,,\,f_\ve^\kappa(\Hf)]\,=\,
\RA{iy}\,[f_\ve^\kappa(\Hf)\,,\,\valpha\cdot\V{A}]\,\RA{iy}
\end{equation*}
on $\sD$ and employing
\eqref{eva2001}, \eqref{eva2001b}, and \eqref{clara1a} we obtain the following
estimate, for all $\vp,\psi\in\sD$, and $E\grg k_1$,
\begin{align*}
\big|\SPb{&|\DA|^{1/2}\,\vp}{\HT^{\nu+\gamma}\,f_\ve^\sigma(\Hf)\,
[\SA\,,\,f_\ve^\kappa(\Hf)]\,\HT^{-\nu+\delta}\,f_\ve(\Hf)^{-\kappa+\tau}
\,\psi}\big|
\\
&\klg\,
\int_\RR\Big|
\SPB{\HT^{\nu+\gamma}\,|\DA|^{1/2}\,\vp}{f_\ve^\sigma(\Hf)
\,[f_\ve^\kappa(\Hf)\,,\,\RA{iy}]\,\times
\\
&\qquad\quad\times\,
\HT^{-\nu+\delta}\,f_\ve^{-\kappa+\tau}(\Hf)\,\psi}\Big|
\,\frac{dy}{\pi}
\\
&=\,
\int_\RR\Big|
\SPB{\vp}{|\DA|^{1/2}\,
\RA{iy}\,\Upsilon_{\sigma,\nu+\gamma}(iy)
\,\HT^{\nu+\gamma}\,f_\ve^\sigma(\Hf)\,
[f_\ve^\kappa(\Hf)\,,\,\valpha\cdot\V{A}]\,\times
\\
&\qquad\quad\times
\,\HT^{-\nu+\delta}\,
f_\ve^{-\kappa+\tau}(\Hf)\,
\wt{\Upsilon}_{\kappa-\tau,\nu-\delta}(iy)\,\RA{iy}\,\psi}
\Big|\,\frac{dy}{\pi}
\\
&\klg\,
C_{\kappa,\nu}\,(d_1+d_\rho)\,E^{\gamma+\delta+\sigma+\tau-1/2}
\,\sup_{y\in\RR}\{\|\Upsilon_{\sigma,\nu+\gamma}(iy)\|\,
\|\wt{\Upsilon}_{\kappa-\tau,\nu-\delta}(iy)\|\}
\\
&\qquad\cdot
\Big(\int_\RR\big\|\,|\DA|^{1/2}\,\RA{iy}\,\vp\,\big\|^2\,
\frac{dy}{\pi}\Big)^{1/2}
\Big(\int_\RR\big\|\,\RA{iy}\,\psi\,\big\|^2\,
\frac{dy}{\pi}\Big)^{1/2}
\\
&\klg\,C_{\kappa,\nu,d_1,d_\rho}\,E^{\gamma+\delta+\sigma+\tau-1/2}
\,\|\vp\|\,\|\psi\|\,.
\end{align*}
This estimate shows that the vector in the right entry of
the scalar product in the first line belongs to
$\dom((|\DA|^{1/2})^*)=\dom(|\DA|^{1/2})$ and that \eqref{clara1b}
holds true.
Next, we observe that \eqref{clara1c} follows from \eqref{clara1b}
and \eqref{maria}.
Finally, \eqref{clara1bb} follows from 
$\|X\|\klg\const(\nu,\kappa,d_1,d_\rho)$, where
$X:=\HT^\nu\,f_\ve^\kappa(\Hf)\,[\SA\,,\,
\HT^{-\nu}\,f_\ve^{-\kappa}(\Hf)]$.
Such a bound on $\|X\|$ is, however, an immediate
consequence of \eqref{clara1b} (where we can choose $\ve=0$)
because
$$
X\,=\,[\HT^{\nu}\,,\,\SA]\,\HT^{-\nu}\,+\,
\HT^{\nu}\,[f_\ve^\kappa(\Hf)\,,\,\SA]\,f_\ve^{-\kappa}(\Hf)\,\HT^{-\nu}
$$
on the domain $\sD$.
\end{proof}


\section{Non-relativistic QED}
\label{sec-nr}

\noindent
The Pauli-Fierz operator for a molecular system
with static nuclei and
$N\in\NN$ electrons interacting with the
quantized radiation field is acting in the Hilbert space
\begin{equation}\label{def-HRN}
\HR_N\,:=\,\cA_N L^2\big((\RR^3\times\ZZ_4)^N\big)\otimes\Fock\,,
\end{equation}
where $\cA_N=\cA_N^2=\cA_N^*$ denotes anti-symmetrization,
$$
(\cA_N\,\Psi)(X)\,:=\,
\frac{1}{N!}\sum_{\pi\in\mathfrak{S}_N}(-1)^\pi\,
\Psi(\V{x}_{\pi(1)},\vs_{\pi(1)},\ldots,\V{x}_{\pi(N)},\vs_{\pi(N)})\,,
$$
for $\Psi\in L^2((\RR^3\times\ZZ_4)^N)$
and a.e. $X=(\V{x}_i,\vs_i)_{i=1}^N\in(\RR^3\times\ZZ_4)^N$.
\`{A} priori it is defined on the dense domain
$$
\sD_N\,:=\,
\cA_N C_0^\infty\big((\RR^3\times\ZZ_4)^N\big) \otimes \sC_0\,,
$$
the tensor product understood in the algebraic sense, by
\begin{equation}\label{def-nrPF}
H_\nr^V\,
\equiv\,H_\nr^V(\V{G})\,:=\,\sum_{i=1}^N(\DA^{(i)})^2\,+\,V\,+\,\Hf\,.
\end{equation}
A superscript $(i)$ indicates that the operator below
is acting on the pair of variables $(\V{x}_i,\vs_i)$. 
In fact, the operator defined in \eqref{def-nrPF} is a 
two-fold copy of the usual Pauli-Fierz operator which acts
on two-spinors and the energy has been shifted by $N$ in \eqref{def-nrPF}.
For \eqref{Clifford} implies
\begin{equation}\label{DA2}
\DA^2=
\cT_\V{A}\oplus\cT_\V{A}\,,\quad
\cT_\V{A}:=\big(\vsigma\cdot(-i\nabla_\V{x}+\V{A})\big)^2+1
\,.
\end{equation}
Here $\vsigma=(\sigma_1,\sigma_2,\sigma_3)$ is a vector
containing the Pauli matrices 
(when $\alpha_j$, $j\in\{0,1,2,3\}$, are given in Dirac's
standard representation).
We write $H_\nr^V$ in the form \eqref{def-nrPF}
to maintain a unified notation throughout this paper.

We shall only make use of the following properties of the
potential $V$.

\begin{hypothesis}\label{hyp-V}
$V$ can be written as
$V=V_+-V_-$, where
$V_\pm\grg0$ is a symmetric operator acting in
$\cA_NL^2\big((\RR^3\times\ZZ_4)^4\big)$
such that $\sD_N\subset\dom(V_\pm)$. 
There exist $a\in(0,1)$ and $b\in(0,\infty)$
such that
$
V_-\klg a\,H_\nr^0+b
$
in the sense of quadratic forms on $\sD_N$. 
\end{hypothesis}

\begin{example}
The Coulomb potential generated by $K\in\NN$ fixed
nuclei located at the positions $\{\V{R}_1,\ldots,\V{R}_K\}\subset\RR^3$
is given as
\begin{equation}\label{def-VC}
\VC(X)\,:=\,-\sum_{i=1}^N\sum_{k=1}^K\frac{e^2\,Z_k}{|\V{x}_i-\V{R}_k|}\,+\,
\sum_{{i,j=1\atop i<j}}^N\frac{e^2}{|\V{x}_i-\V{x}_j|}\,,
\end{equation}
for some $e,Z_1,\ldots,Z_K>0$
and a.e. $X=(\V{x}_i,\vs_i)_{i=1}^N\in(\RR^3\times\ZZ_4)^N$.
It is well-known that $\VC$ is infinitesimally
$H_\nr^0$-bounded and that $\VC$ fulfills
Hypothesis~\ref{hyp-V}.\hfill$\diamond$
\end{example}

\smallskip

\noindent
It follows immediately from Hypothesis~\ref{hyp-V}
that $H_\nr^V$ has a self-adjoint Friedrichs extension
-- henceforth denoted by the same symbol $H_\nr^V$ --
and that $\sD_N$ is a form core for $H_\nr^V$.
Moreover, we have
\begin{equation}\label{fred5}
(\DA^{(1)})^2,\ldots,\,
(\DA^{(N)})^2,\,V_+,\,\Hf\,\klg\,H_\nr^{V_+}\,\klg\,
(1-a)^{-1} \,(H_\nr^V+b)
\end{equation}
on $\sD_N$.
In \cite{FGS2001} it is shown that 
$\dom((H_\nr^V)^{n/2})\subset\dom(\Hf^{n/2})$,
for every $n\in\NN$. We re-derive this result by means of
Corollary~\ref{cor-hoe-allg} in the next theorem where
\begin{align*}
E_{\nr}\,&:=\,\inf\spec[H_\nr^V]\,,
\qquad H_\nr':=\,H_\nr^V-E_\nr+1\,.
\end{align*}

\begin{theorem}\label{thm-hoe-nr}
Assume that $\omega$ and $\V{G}$ fulfill Hypothesis~\ref{hyp-G}
and that $V$ fulfills Hypothesis~\ref{hyp-V}.
Assume in addition that
\begin{align}
\label{hyp-rotGnr}
2\int\omega(k)^{\ell}\,\|\nabla_{\V{x}}\wedge\V{G}(k)\|^2_\infty\,dk
\,&\klg\,d_{\ell+2}^2\,,
\\
\label{hyp-divG}
\int\omega(k)^{\ell}\,\|\nabla_{\V{x}}\cdot\V{G}(k)\|^2_\infty\,dk
\,&\klg\,d_{\ell+2}^2\,,
\end{align}
for all $\ell\in\{-1,0,1,2,\ldots\}$.
Then, for every $n\in\NN$, we have 
$\dom((H_\nr^V)^{n/2})\subset\dom(\Hf^{n/2})$,
$\Hf^{n/2}\,(H_\nr')^{-n/2}$ maps $\dom(H_\nr^V)$ into
itself, and
\begin{align*}
\big\|\,\Hf^{n/2}\,(H_\nr')^{-n/2}\,\big\|\,&\klg\,
C(N,n,a,b,d_{-1},d_1,d_{5+n})\,(|E_\nr|+1)^{(3n-2)/2}\,,
\\
\big\|\,[\Hf^{n/2}\,,\,H_\nr^V]\,(H_\nr')^{-n/2}\,\big\|
\,&\klg\,C'(N,n,a,b,d_{-1},d_1,d_{5+n})\,(|E_\nr|+1)^{(3n-1)/2}\,.
\end{align*}
\end{theorem}

\begin{proof}
We pick the function $f_\ve$ defined in \eqref{clara1} 
with $E=1$ and verify that the operators
$F_\ve^n\,:=\,f^{n/2}_\ve(\Hf)$, $\ve>0$, $n\in\NN$, and $H_\nr'$
fulfill the conditions (a), (b), and (c') of Theorem~\ref{thm-hoe-allg}
and Corollary~\ref{cor-hoe-comm} with $m=\infty$. 
Then the assertion follows from
Corollary~\ref{cor-hoe-allg}. We set $\HT:=\Hf+E$ in what follows.
By means of \eqref{fred5} we find
\begin{equation}\label{daphne0}
\SPn{\Psi}{F_\ve^2\,\Psi}\klg\SPn{\Psi}{\HT\,\Psi}
\klg \,\frac{E_\nr+b+E}{1-a}\,\SPn{\Psi}{H_\nr'\,\Psi}\,,
\end{equation}
for all $\Psi\in\sD_N$,
which is Condition~(b). 
Next, we observe that $F_\ve$ maps $\sD_N$ into itself.
Employing \eqref{fred5} once more and using $-V_-\klg0$
and the fact that $V_+\grg0$ and $F_\ve$ act on different tensor
factors we deduce that
\begin{align}
\SPb{F_\ve\,\Psi}{(V+\Hf)\,F_\ve\,\Psi}
\,&\klg\,\|f_\ve\|_\infty\,\nonumber
\SPb{\Psi}{(V_++\Hf)\,\Psi}
\\
&\klg\,\label{daphne1}\|f_\ve\|_\infty\,
\frac{E_\nr+b+E}{1-a}\,\SPn{\Psi}{H_\nr'\,\Psi}\,,
\end{align}
for every $\Psi\in\sD_N$.
Thanks to \eqref{clara1a} with $\kappa=1/2$, 
$\nu=\gamma=\delta=\sigma=\tau=0$,
and \eqref{fred5} we further find some $C\in(0,\infty)$
such that
\begin{align}
\big\|\,\DA^{(i)}\,F_\ve\,\Psi\,\big\|^2
&\klg\,\nonumber
2\,\|f_\ve\|_\infty\,\|\DA^{(i)}\,\Psi\|^2
+2\,\|f_\ve\|_\infty\,
\big\|\,F_\ve^{-1}\,[\valpha\cdot\V{A}\,,\,F_\ve]\,\big\|^2\,\|\Psi\|^2
\\ \label{daphne2}
&\klg\,C\,\|f_\ve\|_\infty\,\SPn{\Psi}{H_\nr'\,\Psi}\,,
\end{align}
for all $\Psi\in\sD_N$.
\eqref{daphne1} and \eqref{daphne2} together show that
Condition~(a) is fulfilled, too.
Finally, we verify the bound in (c').
We use 
$$
[\valpha\cdot(-i\nabla_\V{x})\,,\,\valpha\cdot\V{A}]\,=\,\vSigma\cdot\V{B}
-i\,(\nabla_\V{x}\cdot\V{A})\,,
$$ 
where
$
\V{B}:=\ad(\nabla_\V{x}\wedge\V{G})+a(\nabla_\V{x}\wedge\V{G})
$
is the magnetic field and the $j$-th entry of the formal
vector $\vSigma$ is $-i\,\epsilon_{jk\ell}\,\alpha_k\,\alpha_\ell$,
$j,k,\ell\in\{1,2,3\}$,
to write the square of the Dirac operator on the domain $\sD$ as
$$
\DA^2\,=\,\DO^2+
\vSigma\cdot\V{B}-i\,(\nabla_\V{x}\cdot\V{A})
+(\valpha\cdot\V{A})^2+2\,\valpha\cdot\V{A}\,\valpha\cdot(-i\nabla_\V{x})\,.
$$
This yields
\begin{align*}
[H_\nr',F_\ve^n]&=
\sum_{i=1}^N\big[(\DA^{(i)})^2,F_\ve^n\big]
=\sum_{i=1}^N\big\{\,
[\vSigma\cdot\V{B}^{(i)}\,,\,F_\ve^n]-i\,
[(\nabla_\V{x}\cdot\V{A}^{(i)})\,,\,F_\ve^n]
\\
&\;
+\valpha\cdot\V{A}^{(i)}\,[\valpha\cdot\V{A}^{(i)}\,,\,F_\ve^n]
+[\valpha\cdot\V{A}^{(i)}\,,\,F_\ve^n]\,(2\DA^{(i)}
-\valpha\cdot\V{A}^{(i)}-2\beta)
\,\big\}
\end{align*}
on $\sD_N$. For every $i\in\{1,\ldots,N\}$, we further write
\begin{align*}
[\valpha\cdot\V{A}^{(i)}\,,\,F_\ve^n]\,\DA^{(i)}
\,&=\,
Q_{\ve,n}^{(i)}\,
\big(\,\DA^{(i)}\,F^{n-1}_\ve
-{Q}_{\ve,n-1}^{(i)}\,F_\ve^{n-2}\,\big)
\end{align*}
on $\sD_N$, where
\begin{align}\label{def-Qveni}
Q_{\ve,n}^{(i)}\,&:=\,[\valpha\cdot\V{A}^{(i)}\,,\,F_\ve^n]\,F_\ve^{1-n},
\quad n\in\NN\,,\qquad Q_{\ve,0}^{(i)}\,:=\,0\,.
\end{align}
According to \eqref{clara1a} we have 
$\|Q_{\ve,n}^{(i)}\|\klg n\,2^{(n+2)/2}\,(d_1+d_{3+n})$,
$\|\HT^{1/2}\,Q_{\ve,n}^{(i)}\,
\HT^{-1/2}\|\klg n\,2^{(n+3)/2} (d_1+d_{4+n})$. 
Likewise, we write
\begin{align*}
[\valpha\cdot\V{A}^{(i)}\,,\,F_\ve^n]\,\valpha\cdot\V{A}^{(i)}
\,&=\,
Q_{\ve,n}^{(i)}\,
\big(\,\{\valpha\cdot\V{A}^{(i)}\,\HT^{-1/2}\}\,\HT^{1/2}\,F^{n-1}_\ve
-{Q}_{\ve,n-1}^{(i)}\,F_\ve^{n-2}\,\big)
\end{align*}
on $\sD_N$, where $\|\valpha\cdot\V{A}\,\HT^{-1/2}\|^2\klg 2\,d_0^2+4\,d_{-1}^2$
by \eqref{rb-aHf}.
Furthermore, we observe that Lemma~\ref{le-clara1} is applicable
to $\vSigma\cdot\V{B}$ as well instead of $\valpha\cdot\V{A}$;
we simply have to replace the form factor $\V{G}$ by $\nabla_\V{x}\wedge\V{G}$
and to notice that $\|\vSigma\cdot\V{v}\|_{\LO(\CC^4)}=|\V{v}|$,
$\V{v}\in\RR^3$, in analogy to \eqref{C*}. Note that the indices
of $d_\ell$ are shifted by $2$ because of \eqref{hyp-rotGnr}.
Finally, we observe that Lemma~\ref{le-clara1} is applicable
to $\nabla_\V{x}\cdot\V{A}$, too. To this end we have to replace
$\V{G}$ by $(\nabla_\V{x}\cdot\V{G},0,0)$ and
$d_\ell$ by some universal constant
times $d_{2+\ell}$ because of \eqref{hyp-divG}.
Taking all these remarks into account we arrive at
\begin{align*}
\big|&\SPb{\Psi_1}{[H_\nr'\,,\,F_\ve^n]\,\Psi_2}\big|
\,\klg\,
\sum_{i=1}^N\Big\{\,\|\Psi_1\|\,
\big\|\,[\vSigma\cdot\V{B}^{(i)}\,,\,F_\ve^n]\,F_\ve^{1-n}\,\big\|
\,\|F_\ve^{n-1}\,\Psi_2\|
\\
&\:+\|\Psi_1\|\,
\big\|\,[\mathrm{div}\,\V{A}^{(i)}\,,\,F_\ve^n]\,F_\ve^{1-n}\,\big\|
\,\|F_\ve^{n-1}\,\Psi_2\|
\\
&\:+\|\Psi_1\|\,\|\valpha\cdot\V{A}\,\HT^{-1/2}\|\,
\big\|\,\HT^{1/2}\,Q_{\ve,n}^{(i)}\,
\HT^{-1/2}\,\big\|
\,\|\HT^{1/2}\,F_\ve^{n-1}\,\Psi_2\|
\\
&\:+
\|\Psi_1\|\,\|Q_{\ve,n}^{(i)}\|\,\big(\,2\,
\|\DA^{(i)}\,F_\ve^{n-1}\,\Psi_2\|+
\|\valpha\cdot\V{A}\,\HT^{-1/2}\|\,
\|\HT^{1/2}\,F_\ve^{n-1}\,\Psi_2\|\,\big)
\\
&\:+3\,\|\Psi_1\|\,\|Q_{\ve,n}^{(i)}\|\,
\|{Q}_{\ve,n-1}^{(i)}\|\,\|F_\ve^{n-2}\,\Psi_2\|
+2\,\|\Psi_1\|\,\|Q_{\ve,n}^{(i)}\|\,\|\beta\|\,\|F_\ve^{n-1}\,\Psi_2\|\,\Big\}\,,
\end{align*}
for all $\Psi_1,\Psi_2\in\sD_N$.
From this estimate, Lemma~\ref{le-clara1}, and \eqref{fred5}
we readily infer that Condition~(c') is valid
with $c_n= (|E_\nr|+1)\,C''(N,n,a,b,d_{-1},\ldots,d_{5+n})$.
\end{proof}


\section{The semi-relativistic Pauli-Fierz operator}
\label{sec-PF}

\noindent
The semi-relativistic Pauli-Fierz operator is also acting
in the Hilbert space $\HR_N$ introduced in \eqref{def-HRN}.
It is obtained by substituting the non-local
operator $|\DA|$ for $\DA^2$ in $H_\nr^V$.
We thus define, \`{a} priori on the dense domain $\sD_N$,
$$
H_\sr^V\,\equiv\,
H_\sr^V(\V{G})\,:=\,
\sum_{i=1}^N|\DA^{(i)}|\,+\,V\,+\,\Hf\,,
$$
where $V$ is assumed to fulfill Hypothesis~\ref{hyp-V}
with $H_\nr^0$ replaced by $H_\sr^0$.
To ensure that in the case of the Coulomb potential
$\VC$ defined in \eqref{def-VC} 
this yields a well-defined self-adjoint operator
we have to impose appropriate restrictions on the nuclear
charges. 

\begin{example}\label{ex-C-srPF}
In Proposition~\ref{prop-sb-srPF} we show that $H_\sr^{\VC}$
is semi-bounded below on $\sD_N$ provided that $Z_k\in(0,2/\pi e^2]$,
for all $k\in\{1,\ldots,K\}$.
Its proof is actually a straightforward
consequence of \eqref{maria} and a commutator estimate
obtained in \cite{MatteStockmeyer2009a}. 
If all atomic numbers $Z_k$ are strictly less than $2/\pi e^2$
we thus find $a\in(0,1)$ and $b\in(0,\infty)$ such that
\begin{equation}\label{rfb-C-srPF}
\sum_{i=1}^N\sum_{k=1}^K\frac{e^2\,Z_k}{|\V{x}_i-\V{R}_k|}\,\klg\,
a\,H_\sr^0+b
\end{equation}
in the sense of quadratic forms on $\sD_N$. In particular, $\VC$
fulfills Hypothesis~\ref{hyp-V}
with $H_\nr^0$ replaced by $H_\sr^0$ as long as
$Z_k\in(0,2/\pi e^2)$, for $k\in\{1,\ldots,K\}$.
\hfill$\diamond$
\end{example}

\smallskip

\noindent
For potentials $V$ as above
$H_\sr^V$ has a self-adjoint Friedrichs extension
which we denote again by the same symbol $H_\sr^V$.
Moreover, $\sD_N$ is a form core for $H_\sr^V$ and 
we have the following analogue of \eqref{fred5},
\begin{equation} \label{manuela0}
|\DA^{(1)}|,\ldots,|\DA^{(N)}|,\,V_+,\,\Hf\,\klg\,
H_\sr^{V_+}\,\klg\,(1-a)^{-1}\,(H_\sr^V+b)
\end{equation}
on $\sD_N$.
In order to apply Corollary~\ref{cor-hoe-allg} to the
semi-relativistic Pauli-Fierz operator we recall
the following special case of \cite[Corollary~3.7]{KMS2009a}:

\begin{lemma}
Assume that $\omega$ and $\V{G}$ fulfill
Hypothesis~\ref{hyp-G}.
Let $\tau\in(0,1]$.
Then there exist $\delta>0$ and $C\equiv C(\delta,\tau,d_1)\in(0,\infty)$
such that
\begin{equation}\label{vgl-Gi}
C+|\D{\V{A}}|+\tau\,\Hf\grg\delta\,(|\DO|+\Hf)
\grg\delta\,(|\DO|+\tau\,\Hf)
\,\grg\,\delta^2\,|\DA|-\delta\,C
\end{equation}
in the sense of quadratic forms on $\sD$.
\end{lemma}

\smallskip

\noindent
In the next theorem  we re-derive the higher order estimates
obtained in \cite{FGS2001} for the semi-relativistic 
Pauli-Fierz operator by means of Corollary~\ref{cor-hoe-allg}.
(The second estimate of Theorem~\ref{thm-hoe-srPF} is actually slightly
stronger than the corresponding one stated in \cite{FGS2001}.) 
The estimates of the following proof are also employed in 
Section~\ref{sec-np} where we treat the no-pair operator.
We set
$$
E_\sr\,:=\,\inf\spec[H_\sr]\,,
\qquad H_\sr'\,:=\,H_\sr^V-E_\sr+1\,.
$$

\begin{theorem}\label{thm-hoe-srPF}
Assume that $\omega$ and $\V{G}$ fulfill Hypothesis~\ref{hyp-G}
and that $V$ fulfills Hypothesis~\ref{hyp-V} with
$H_\nr^0$ replaced by $H_\sr^0$.
Then, for every
$m\in\NN$, it follows that 
$\dom((H_\sr^V)^{m/2})\subset \dom(\Hf^{m/2})$,
$\Hf^{m/2}\,(H_\sr')^{-m/2}$ maps $\dom(H_\sr^V)$ into itself, and
\begin{align*}
\big\|\,\Hf^{m/2}\,(H_\sr')^{-m/2}\,\big\|\,&\klg\,
C(N,m,a,b,d_{1},d_{3+m})\,(|E_\sr|+1)^{(3m-2)/2}\,,
\\
\big\|\,[\Hf^{m/2}\,,\,H_\sr^V]\,(H_\sr')^{-m/2}\,\big\|
\,&\klg\,C'(N,m,a,b,d_1,d_{3+m})\,(|E_\sr|+1)^{(3m-1)/2}
\,.
\end{align*}
\end{theorem}

\begin{proof}
Let $m\in\NN$.
We pick the function $f_\ve$ defined in \eqref{clara1}
with $E=k_1\vee C$. ($k_1$ is the constant appearing in
Lemma~\ref{le-clara1b} with $\kappa=m/2$, $\nu=0$, and depends
on $m$, $d_1$, and $d_{3+m}$; $C$ is the one in \eqref{vgl-Gi}.)
We fix some $n\in\NN$, $n\klg m$,
and verify Conditions~(a), (b), and (c') of Theorem~\ref{thm-hoe-allg}
and Corollary~\ref{cor-hoe-comm}
with $F_\ve=f_\ve^{1/2}(\Hf)$, $\ve>0$.
The estimates \eqref{daphne0} and \eqref{daphne1}
are still valid without any further change
when the subscript $\nr$ is replaced by $\sr$. 
Employing \eqref{vgl-Gi} twice and using \eqref{manuela0}
we obtain the following substitute of \eqref{daphne2},
\begin{align*}
\SPb{F_\ve\,&\Psi}{|\DA|\,F_\ve\,\Psi}
\klg \delta^{-1}
\|\,|\DO|^{1/2}\,F_\ve\,\Psi\|^2+\delta^{-1}\,\|\HT^{1/2}\,F_\ve\,\Psi\|^2
\\
&\klg\,\delta^{-1}\,
\|f_\ve\|_\infty\,\big(\|\,|\DO|^{1/2}\,\Psi\|^2+\|\HT^{1/2}\,\Psi\|^2\big)
\klg
C'\,\|f_\ve\|_\infty\,
\SPb{\Psi}{H_\sr'\,\Psi}\,,
\end{align*}
for all $\Psi\in\sD_N$. Altogether we see that Conditions~(a) and~(b)
are satisfied.
In order to verify (c') we set 
\begin{equation}\label{def-Uveni}
U_{\ve,n}^{(i)}\,:=\,
[\SA^{(i)}\,,\,F_\ve^n]\,
F_\ve^{1-n}\,=\,F_\ve^n\,[F_\ve^{-n}\,,\,\SA^{(i)}]\,F_\ve\,,
\qquad i\in\{1,\ldots,N\}\,.
\end{equation}
By virtue of \eqref{clara1b} we know that the norms
of $U_{\ve,n}^{(i)}$ and $U_{\ve,n}^{(i)}\,|\DA^{(i)}|^{1/2}$ are
bounded uniformly in $\ve>0$ by some constant,
$C\in(0,\infty)$, that depends only
on $n$, $d_1$, and $d_{3+n}$.
We employ the notation \eqref{def-Qveni} and \eqref{def-Uveni}
to write
\begin{align*}
&[H_\sr'\,,\,F_\ve^n]\,=\,\nonumber
\sum_{i=1}^N\big[\,|\DA^{(i)}|\,,\,F_\ve^n\,\big]
\,=\,\sum_{i=1}^N\big[\,\SA^{(i)}\,\DA^{(i)}\,,\,F_\ve^n\,\big]
\\
&=
\sum_{i=1}^N\Big\{
\{U_{\ve,n}^{(i)}\,|\DA^{(i)}|^{1/2}\}\,\SA^{(i)}\,|\DA^{(i)}|^{1/2}\,F_\ve^{n-1}
-U_{\ve,n}^{(i)}\,Q_{\ve,n-1}^{(i)}\,F_\ve^{n-2}
+\SA^{(i)}\,Q_{\ve,n}^{(i)}\,F_\ve^{n-1}\Big\}.
\end{align*}
The previous identity, \eqref{manuela0}, and
$|\DA|\grg1$ permit to get 
\begin{align*}
\big|\SPb{\Psi_1}{[H_\sr'\,,\,F_\ve^n]\,\Psi_2}\big|
\,&\klg\,
\sum_{i=1}^N\|\Psi_1\|\,\big\{
C\,\big\|\,|\DA^{(i)}|^{1/2}\,F_\ve^{n-1}\,\Psi_2\big\|
\\
&\quad+C\,\|Q_{\ve,n}^{(i)}\|\,\|F_\ve^{n-2}\,\Psi_2\|
+\|Q_{\ve,n}^{(i)}\|\,\|F_\ve^{n-1}\,\Psi_2\|\big\}
\\
&\klg\,
c_n\,\big\{\,\|\Psi_1\|^2+\SPb{F_\ve^{n-1}\,\Psi_2}{H_\sr'\,
F_\ve^{n-1}\,\Psi_2}\,\big\}
\,,
\end{align*}
for all $\Psi_1,\Psi_2\in\sD_N$ and some 
constant $c_n=C''(n,a,b,d_1,d_{3+n})\,(|E_\sr|+1)$.
So (c') is fulfilled also and the assertion follows from
Corollary~\ref{cor-hoe-allg}.
\end{proof}


\section{The no-pair operator}
\label{sec-np}

\noindent
We introduce the spectral projections
\begin{equation}\label{def-PA}
\PA\,:=\,E_{[0,\infty)}(\DA)\,=\,
\frac{1}{2}\,\id+\,\frac{1}{2}\,\SA\,,
\qquad \PAm\,:=\,\id-\PA\,.
\end{equation}
The no-pair operator acts in the projected Hilbert
space 
$$
\HR_N^+\,\equiv\,\HR_N^+(\V{G})\,:=\,\PAN\,\HR_N\,,\qquad
\PAN\,:=\,
\prod_{i=1}^NP_\V{A}^{+,(i)}\,,
$$
and is \`{a} priori defined on the dense domain
$\PAN\,\sD_N$ by
\begin{equation*}
H_\np^V\,\equiv\,H_\np^V(\V{G})\,:=\,
\PAN\,
\Big\{\,\sum_{i=1}^N\DA^{(i)}+V+\Hf\,\Big\}\,\PAN\,.
\end{equation*}
Notice that all operators $\DA^{(1)},\ldots,\DA^{(N)}$ and 
$\Pa{1},\ldots,\Pa{N}$
commute in pairs owing to the fact that the components
of the vector potential $A^{(i)}(\V{x})$, $A^{(j)}(\V{y})$,
$\V{x},\V{y}\in\RR^3$, $i,j\in\{1,2,3\}$,
commute in the sense that all their spectral
projections commute; see the appendix to
\cite{LiebLoss2002} for more details. 
(Here we use the assumption that 
$\V{G}_{\V{x}}(-\V{k},\lambda)=\ol{\V{G}_{\V{x}}(\V{k},\lambda)}$.)
So the order of the application
of the projections $\Pa{i}$ is immaterial.
In this section we restrict the discussion 
to the case where $V$ is given by the Coulomb potential $\VC$
defined in \eqref{def-VC}. To have a handy notation we set
$$
v_i\,:=\,-\sum_{k=1}^K\frac{e^2\,Z_k}{|\V{x}_i-\V{R}_k|}\,,
\qquad w_{ij}\,:=\,\frac{e^2}{|\V{x}_i-\V{x}_j|}\,,
$$
for all $i\in\{1,\ldots,N\}$ and $1\klg i<j\klg N$,
respectively.
Thanks to \cite[Lemma~3.4(ii)]{MatteStockmeyer2009a},
which implies that $\PA$ maps $\sD$ into 
$\dom(|\DO|)\cap\dom(\Hf^\nu)$, for every $\nu>0$,
and Hardy's inequality,
we actually know that $H_\np^{\VC}$ is well-defined on $\sD_N$.
In order to apply Corollary~\ref{cor-hoe-allg} to $H_\np^{\VC}$
we extend $H_\np^{\VC}$ to a continuously invertible operator
on the whole space $\HR_N$: We pick the 
complementary projection,
$$
\PANb\,:=\,
\id-\PAN\,,
$$
abbreviate
$$
\Pa{i,j}\,:=\,\Pa{i}\,\Pa{j}\,=\,\Pa{j}\,\Pa{i}\,,\qquad
1\klg i<j\klg N\,,
$$
and define the operator $\wt{H}_\np$
\`{a} priori on the domain $\sD_N$ by
\begin{align}
\wt{H}_\np\,&:=\,\nonumber
\sum_{i=1}^N\big\{\,|\DA^{(i)}|+\Pa{i}\,v_i\,\Pa{i}\,\big\}
+\sum_{{i,j=1\atop i<j}}^N
\Pa{i,j}\,w_{ij}\,\Pa{i,j}
\\ \label{def-Hnp'}
&\quad+\PAN\,\Hf\,\PAN
+\PANb\,\Hf\,\PANb\,.
\end{align}
Evidently, we have $[\wt{H}_\np\,,\,\PAN]=0$ 
and $\wt{H}_\np\,\PAN=H_\np^{\VC}\,\PAN$ on $\sD_N$.
In Proposition~\ref{prop-sb-np} we show that the
quadratic forms of the no-pair operator
$H_\np^{\VC}$ and of $\wt{H}_\np$ are
semi-bounded below on $\sD_N$ provided that the atomic
numbers $Z_1,\ldots,Z_K\grg0$ are less than
the critical one of the Brown-Ravenhall
model determined in \cite{EPS1996},
\begin{equation}\label{def-Znp}
Z_\np\,:=\,(2/e^2)/(2/\pi+\pi/2)\,.
\end{equation}
Therefore, both $H_\np^{\VC}$ and $\wt{H}_\np$ possess
self-adjoint Friedrichs extensions which are again denoted
by the same symbols in the sequel.
$\sD_N$ is a form core for $\wt{H}_\np$ and we have the bound
\begin{equation}\label{gisela}
\wt{H}_\np-\sum_{i=1}^N\Pa{i}\,v_i\,\Pa{i}\,
\klg\,\frac{Z_\np+|\sZ|}{Z_\np-|\sZ|}\,\big(\wt{H}_\np
+C(N,\sZ,\sR,d_{-1},d_1,d_5)\big)
\end{equation}
on $\sD_N$,
where $|\sZ|:=\max\{Z_1,\ldots,Z_K\}<Z_\np$.
Moreover, it makes sense to define
$$
E_\np\,:=\,\inf\spec[H_\np^{\VC}]\,,
$$ 
so that
$$
H_\np'\,:=\,
\wt{H}_\np-E_\np\,\PAN+\id\,\grg\,\id\,.
$$

\begin{theorem}\label{thm-hoe-np}
Assume that $\omega$ and $\V{G}$ fulfill
Hypothesis~\ref{hyp-G} and let
$N,K\in\NN$, $e>0$,
$\sZ=(Z_1,\ldots,Z_K)\in[0,Z_\np)^K$, 
and $\sR=\{\V{R}_1,\ldots,\V{R}_K\}\subset\RR^3$,
where $Z_\np$ is defined in \eqref{def-Znp}. 
Then $\dom((H_\np')^{m/2})\subset\dom(\Hf^{m/2})$,
for every $m\in\NN$, and
\begin{align*}
&\big\|\,\Hf^{m/2}\!\!\upharpoonright_{\HR_N^+}
(H_\np-(E_\np-1)\,\id_{\HR_N^+})^{-m/2}\,\big\|_{\LO(\HR_N^+,\HR_N)}
\,\klg\,\big\|\,\Hf^{m/2}\,(H_\np')^{-m/2}\,\big\|
\\
&\quad\klg\,C(N,m,\sZ,\sR,e,d_{-1},d_1,d_{5+m})\,(1+|E_\np|)^{(3m-2)/2}
<\infty\,.
\end{align*}
\end{theorem}

\begin{proof}
Let $m\in\NN$.
Again we pick the function $f_\ve$ defined in \eqref{clara1}
and set $F_\ve:=f_\ve^{1/2}(\Hf)$, $\ve>0$.
This time we
choose $E=\max\{k\,d_1^2,k_1,C\}$ where $k$ is the constant
appearing in \eqref{maria}, 
$C\equiv C(d_1)$ is the one in \eqref{vgl-Gi}, and 
$k_1$ the one appearing in Lemma~\ref{le-clara1b}
with $|\kappa|=(m+1)/2$, $|\nu|=1/2$.
Thus $k_1$ depends only on $m$, $d_1$, and $d_{5+m}$.
On account of Corollary~\ref{cor-hoe-allg} it suffices to show
that the conditions (a)--(c) of Theorem~\ref{thm-hoe-allg}
are fulfilled.
To this end we observe that on $\sD_N$ the extended
no-pair operator can be written as
$H_\np'=H_\sr^0+\id+W$, where
\begin{align*}
W\,&:=\,
\sum_{i=1}^N\Pa{i}\,v_i\,\Pa{i}+\sum_{{i,j=1\atop i<j}}^N
\Pa{i,j}\,w_{ij}\,\Pa{i,j}
\\
&\qquad
-E_\np\,\PAN
-2\Re\big[\PAN\,\Hf\,\PANb\big]\,.
\end{align*}
The
semi-relativistic Pauli-Fierz operator
$H_\sr^0$ has already been treated in the previous section
and the bound
\begin{equation}\label{gisela2}
\Hf\,\klg\,2\PAN\,\Hf\,\PAN+2\PANb\,\Hf\,\PANb
\end{equation} 
together with \eqref{gisela} implies
\begin{align}\label{gisela3}
H_\sr^0\,&\klg\,2\,\wt{H}_\np-2\sum_{i=1}^N
\Pa{i}\,v_i\,\Pa{i}\,\klg\,
C'\,(1+|E_\np|)\,H_\np'
\end{align}
on $\sD_N$, for some $C'\equiv C'(N,\sZ,\sR,d_{-1},d_1,d_{5})\in(0,\infty)$.
Hence, it only remains to consider the operator $W$.

We fix some $n\in\NN$, $n\klg m$.
When we verify (a) we can ignore the potentials $v_i$
since they are negative.
Using $[F_\ve^n,\PANb]=[\PAN,F_\ve^n]$ we obtain
\begin{align*}
\big|2\Re\SPb{&\PAN\,F_\ve\,\Psi}{\Hf\,\PANb\,F_\ve\,\Psi}\big|
\\
&\klg\,
\big\|\,\Hf^{1/2}\,\PAN\,F_\ve\,\Psi\,\big\|^2
+\big\|\,\Hf^{1/2}\,\PANb\,F_\ve\,\Psi\,\big\|^2
\\
&\klg\,
2\,\|f_\ve\|_\infty\,\big\|\Hf^{1/2}\,\PAN\,\Psi\big\|^2
+2\,\|f_\ve\|_\infty\,\big\|\Hf^{1/2}\,\PANb\,\Psi\big\|^2
\\
&\quad +\,4\,
\big\|\,\Hf^{1/2}\,[\PAN\,,\,F_\ve]\,
\HT^{-1/2}\,\big\|
\,\|\HT^{1/2}\,\Psi\|^2,
\end{align*}
for every $\Psi\in\sD_N$, where, for all $n\in\NN$ and $\nu\in\RR$,
\begin{align*}
\HT^\nu&\,[\PAN\,,\,F_\ve^n]\,\HT^{-\nu}\,F_\ve^{1-n}
\,=\,
\sum_{i=1}^N\Big\{\prod_{j=1}^{i-1} \HT^\nu\,\Pa{j}\,\HT^{-\nu}\Big\}
\,\times
\\
&\quad\times
\,\big\{\HT^\nu\,[\Pa{i},F_\ve^n]\,\HT^{-\nu}\,F_\ve^{1-n}\big\}
\Big\{\prod_{k=i+1}^N\HT^\nu\,F_\ve^{n-1}\,\Pa{k}\,\HT^{-\nu}
\,F_\ve^{1-n}\Big\}
\end{align*}
on $\sD_N$.
On account of Corollary~\ref{cor-clara1} we thus have, for 
$|\nu|\klg1/2$,
\begin{equation}\label{gisela4}
\sup_{\ve>0}
\big\|\,\Hf^{\nu}\,[\PAN\,,\,F_\ve^n]\,\HT^{-\nu}\,
F_\ve^{1-n}\,\big\|\,\klg\,C(N,n,d_1,d_{4+n})\,.
\end{equation}
Likewise we have
\begin{align}
\big|\SPb{F_\ve\,\Psi}{\Pa{i,j}\,w_{ij}\,\Pa{i,j}\,&F_\ve\,\Psi}\big|
\,\klg\,\nonumber
2\,\|f_\ve\|\,\big\|\,w_{ij}^{1/2}\,\Pa{i,j}\,\Psi\,\big\|^2
\\ \label{diether}
&+4\,
\big\|\,w_{ij}^{1/2}\,[\Pa{i,j}\,,\,F_\ve]
\,\HT^{-1/2}\,\big\|^2\,\|\HT^{1/2}\,\Psi\|^2,
\end{align}
where the first norm in the second line of \eqref{diether} is bounded 
(uniformly in $\ve>0$) due to Lemma~\ref{le-wij}.
Taking these remarks,
$v_i\klg0$, \eqref{gisela}, and \eqref{gisela2}
into account
we infer that
\begin{align*}
\SPb{F_\ve\,\Psi}{H_\np'\,F_\ve\,\Psi}\,\klg\,c_\ve
\SPb{\Psi}{H_\np'\,\Psi}\,,\qquad \Psi\in\sD_N\,,
\end{align*}
showing that (a) is fulfilled.
Condition~(b) with 
$c^2=C(N,\sZ,\sR,d_{-1},d_1,d_5)(1+|E_\np|)$ follows immediately from
$F_\ve^2\klg\HT\klg H_\sr^0+E$ on $\sD_N$ and \eqref{gisela3}.
Finally, we turn to Condition~(c).
To this end let
$\PANs$ and $\PANf$ be $\PAN$ or $\PANb$.
On $\sD_N$ we clearly have
\begin{align}
\big[\,&\PANs\,\Hf\,\PANf\,,\,F_\ve^n\,\big]
\,=\,\label{zita1}
\pm \,[\PAN\,,\,F_\ve^n]\,\Hf\,\PANf\,
\pm\, \PANs\,\Hf\,[\PAN\,,\,F_\ve^n]\,.
\end{align}
For $\Psi_1,\Psi_2\in\sD_N$, we thus obtain
\begin{align}\nonumber
\big|\SPb{&\Psi_1}{\big[\,\PANs\,\Hf\,\PANf\,,\,F_\ve^n\,\big]\,\Psi_2}\big|
\\ \nonumber
&\klg\,\|\HT^{1/2}\,\Psi_1\|\,
\big\|\,\HT^{-1/2}\,[\PAN\,,\,F_\ve^n]\,\Hf^{1/2}F^{1-n}\,\big\|
\,
\big\|\Hf^{1/2}\,F_\ve^{n-1}\,\PANf\,\Psi_2\big\|
\\
&\quad+\|\Hf^{1/2}\,\PANs\,\Psi_1\|\,
\big\|\Hf^{1/2}\,[\PAN\,,\,F_\ve^n]\,
\HT^{-1/2}\,F_\ve^{1-n}\big\|\,\|\HT^{1/2}\,F_\ve^{n-1}\,\Psi_2\|
\,,\label{gisela5a}
\end{align}
where we can further estimate
\begin{align}\nonumber
\big\|\Hf^{1/2}&\,F_\ve^{n-1}\,\PANf\,\Psi_2\big\|
\\
&\klg\,\nonumber
\big\{1+\big\|\Hf^{1/2}\,F_\ve^{n-1}\,\PAN\,
\HT^{-1/2}\,F_\ve^{1-n}\big\|\big\}\,\|\HT^{1/2}\,F_\ve^{n-1}\,\Psi_2\|
\\
&\klg\,
\big\{1+\|\Hf^{1/2}\,F_\ve^{n-1}\,\PA\,
\HT^{-1/2}\,F_\ve^{1-n}\|^N\big\}\,\|\HT^{1/2}\,F_\ve^{n-1}\,\Psi_2\|
\,,\label{gisela5}
\end{align}
and, of course,
\begin{align}\label{gisela6}
\|\HT^{1/2}\,F_\ve^{n-1}\,\Psi_2\|\,\klg\,
\|\HT^{1/2}\,\PAN\,F_\ve^{n-1}\,\Psi_2\|
+\|\HT^{1/2}\,\PANb\,F_\ve^{n-1}\,\Psi_2\|\,.
\end{align}
The operator norms in \eqref{gisela5a} can be estimated by means
of \eqref{gisela4} with $\nu=\pm1/2$, the one in the last line of
\eqref{gisela5} is bounded by some $C(n,d_1,d_{3+n})\in(0,\infty)$
due to \eqref{clara1bb}. 
In a similar fashion we obtain, for all $i,j\in\{1,\ldots,N\}$,
$i<j$, and $\Psi_1,\Psi_2\in\sD_N$,
\begin{align}\nonumber
&\big|\SPb{\Psi_1}{[\Pa{i,j}\,w_{ij}\Pa{i,j}\,,\,F_\ve^n]\,\Psi_2}\big|
\\
&\klg\,\nonumber
\big\|\,F_\ve^{1-n}\,w_{ij}^{1/2}\,
[F_\ve^n\,,\,\Pa{i,j}]\,\HT^{-1/2}\,\big\|
\,\|\HT^{1/2}\,\Psi_1\|\,
\big\|F_\ve^{n-1}\,w_{ij}^{1/2}\,\Pa{i,j}\,\Psi_2\big\|
\\
&\,+\label{gisela7}
\big\|\,w_{ij}^{1/2}\,\Pa{i,j}\,\Psi_1\big\|\,
\big\|\,w_{ij}^{1/2}\,[\Pa{i,j},\,F_\ve^n]\,F_\ve^{1-n}\,\HT^{-1/2}\big\|
\,\|\HT^{1/2}\,F_\ve^{n-1}\,\Psi_2\|.
\end{align}
Here we can further estimate
\begin{align}\nonumber
\big\|\,w_{ij}^{1/2}\,F_\ve^{n-1}&\,\Pa{i,j}\,\Psi_2\big\|\,
\,\klg\,
\big\|\,w_{ij}^{1/2}\,\Pa{i,j}\,F_\ve^{n-1}\,\Psi_2\big\|
\\
&+\,\label{gisela8}
\big\|\,w_{ij}^{1/2}\,[F_\ve^{n-1}\,,\,\Pa{i,j}]\,
\HT^{-1/2}F_\ve^{1-n}\,\big\|\,\|\HT^{1/2}\,F_\ve^{n-1}\,\Psi_2\|
\,.
\end{align}
Lemma~\ref{le-wij} below ensures that
all operator norms in \eqref{gisela7} and \eqref{gisela8}
that involve $w_{ij}^{1/2}$ are bounded uniformly in $\ve>0$
by constants depending only on $e,n,d_1$, and $d_{5+n}$. Furthermore,
it is now clear how to treat 
the terms involving $v_i$ or $E_\np$.
(In order to treat $v_i$ just replace $\Pa{i,j}$ by $\Pa{i}$,
$w_{ij}$ by $v_i$, and $w_{ij}^{1/2}$ by $|v_i|^{1/2}$ in
\eqref{gisela7} and \eqref{gisela8}.)
Combining \eqref{zita1}--\eqref{gisela8} and their analogues
for the remaining operators in $W$ we arrive at
\begin{align*}
&\big|\SPb{\Psi_1}{[W\,,\,F_\ve^n]\,\Psi_2}\big|
\\
&\klg\,
C\!\!\sum_{\sharp\in\{+,\bot\}}
\big\{\SPb{\Psi_1}{\PANs\,\Hf\,\PANs\,\Psi_1}+
\SPb{F_\ve^{n-1}\,\Psi_2}{\PANs\,\Hf\,\PANs\,F_\ve^{n-1}\,\Psi_2}\big\}
\\
&+C\!\sum_{{i,j=1\atop i<j}}^N\big\{
\SPb{\Psi_1}{\Pa{i,j}w_{ij}\Pa{i,j}\Psi_1}
+
\SPb{F_\ve^{n-1}\Psi_2}{\Pa{i,j}w_{ij}\Pa{i,j}F_\ve^{n-1}\Psi_2}
\big\}
\\
&+C\sum_{i=1}^N
\big\{
\SPb{\Psi_1}{\Pa{i}\,|v_{i}|\,\Pa{i}\,\Psi_1}
+
\SPb{F_\ve^{n-1}\,\Psi_2}{\Pa{i}\,|v_i|\,\Pa{i}\,F_\ve^{n-1}\,\Psi_2}
\big\}
\\
&+C\,(1+|E_\np|)\,\big\{\|\Psi_1\|^2+\|F_\ve^{n-1}\,\Psi_2\|^2\big\}\,,
\end{align*}
for all $\Psi_1,\Psi_2\in\sD_N$ and some $\ve$-independent 
$C\equiv C(N,n,e,d_1,d_{5+n})\in(0,\infty)$.
Employing successively \eqref{maria}, which implies
$|v_i|\klg(\pi e^2|\sZ|/2)(|\DA^{(i)}|+\HT)$, 
after that \eqref{clara1bb}, which yields
$\|\HT^{1/2}\,\Pa{i}\,\Psi\|^2\klg C(d_1,d_4)(\|\HT^{1/2}\,\PAN\,\Psi\|^2
+\|\HT^{1/2}\,\PANb\,\Psi\|^2)$,
and finally \eqref{gisela}
we conclude that 
Condition~(c) is fulfilled
with $c_n=C(N,n,\sZ,\sR,e,d_{-1},d_1,d_{5+n})(1+|E_\np|)$.
\end{proof}

\begin{lemma}\label{le-wij}
For all $i,j\in\{1,\ldots,N\}$, $i<j$, $n\in\ZZ$, and
$\sigma,\tau\grg0$ with $\sigma+\tau\klg1$, 
\begin{align*}
\sup_{\ve>0}&\big\|F_\ve^{\sigma-n}\,w_{ij}^{1/2}
\,[F_\ve^n\,,\,\Pa{i,j}]\,\HT^{-1/2}\,F_\ve^\tau\,\big\|
\\
&=\,
\sup_{\ve>0}\big\|\,w_{ij}^{1/2}\,F_\ve^\sigma
\,[\Pa{i,j}\,,\,F_\ve^{-n}]\,\HT^{-1/2}\,F_\ve^{n+\tau}\,\big\|
\,\klg\,e\,C(n,d_1,d_{5+n})
\,<\,\infty\,.
\end{align*}
\end{lemma}

\begin{proof}
We write
\begin{align*}
&\,w_{ij}^{1/2}\,F_\ve^\sigma\,[\Pa{i}\,\Pa{j}\,,\,F_\ve^{-n}]\,\HT^{-1/2}
F_\ve^{n+\tau}\,=\,
Y_1+w_{ij}^{1/2}\,Y_2+Y_3\,,
\end{align*}
where
\begin{align*}
Y_1\,&:=\,
\{w_{ij}^{1/2}\,F_\ve^{\sigma}\,[\Pa{i}\,,\,F_\ve^{-n}]\,\HT^{-1/2}\,
F_\ve^{n+\tau}\}
\{\HT^{1/2}\,F_\ve^{-n-\tau}\,\Pa{j}\,\HT^{-1/2}\,F_\ve^{n+\tau}\}\,,
\\
Y_2\,&:=\,
\Pa{i}\,F_\ve^{\sigma}\,[\Pa{j}\,,\,F_\ve^{-n}]\,
\HT^{-1/2}\,F_\ve^{n+\tau}\,,
\\
Y_3\,&:=\,
w_{ij}^{1/2}\,[F_\ve^{\sigma}\,,\,\Pa{i}]\,[\Pa{j}\,,\,F_\ve^{-n}]\,
\HT^{-1/2}\,F_\ve^{n+\tau}\,.
\end{align*}
Applying Corollary~\ref{cor-clara1} we immediately see that
$\|Y_1\|\klg e\,C(n,d_1,d_{5+n})$ and that 
$$
\|Y_3\|\klg
\big\|\,w_{ij}^{1/2}\,[F_\ve^{\sigma}\,,\,\Pa{i}]
\,F_\ve^{-\sigma}\big\|\,
\big\|\,F_\ve^{\sigma}\,[\Pa{j}\,,\,F_\ve^{-n}]\,
F_\ve^{n+\tau}\,\big\|
\klg e\,C(n,d_1,d_{3+n})
$$
uniformly in $\ve>0$.
Employing \eqref{maria} (with respect to the variable $\V{x}_j$
for each fixed $\V{x}_i$)
and using $[|\DA^{(j)}|^{1/2},\Pa{i}]=0$,
we further get
\begin{align*}
\big\|\,w_{ij}^{1/2}\,Y_2\,&\Psi\big\|^2\,\klg\,(\pi e^2/2)\,\|\Pa{i}\|^2\,
\big\|\,|\DA^{(j)}|^{1/2}\,F_\ve^{\sigma}\,[\Pa{j}\,,\,F_\ve^{-n}]
\,F_\ve^{n+\tau}\big\|^2\,\|\HT^{-1/2}\|^2
\\
&
+ (\pi e^2/2)\,\big\|\HT^{1/2}\,\Pa{i}\,\HT^{-1/2}\big\|^2\,
\big\|\HT^{1/2}\,F_\ve^{\sigma}\,[\Pa{i}\,,\,F_\ve^{-n}]\,\HT^{-1/2}\,
F_\ve^{n+\tau}\big\|^2.
\end{align*}
By Corollary~\ref{cor-clara1} all norms on the right hand side
are bounded uniformly in $\ve>0$ by constants depending only
on $n,d_1$, and $d_{4+n}$.
\end{proof}


\appendix

\section{Semi-boundedness of $H_\sr^{\VC}$ and $H_\np^{\VC}$}
\label{app-sb}

\noindent
In this appendix we verify that the semi-relativistic
Pauli-Fierz and no-pair operators with Coulomb potential are semi-bounded 
below for all nuclear charges less than the critical
charges without radiation fields.
We do not attempt to give good lower
bounds on their spectra since this is not the topic addressed
in this paper. 
Our aim here is essentially only to ensure that these operators
possess self-adjoint Friedrichs extensions.
We recall that the stability of matter of the second
kind has been proven for the no-pair operator
in \cite{LiebLoss2002} under certain
restrictions on the fine-structure constant,
the ultra-violet cut-off, and the nuclear
charges. 
The stability of matter of the second kind is a much
stronger property than mere semi-boundedness.
It says that the operator is bounded
below by some constant which is proportional
to the total number of nuclei and electrons
and uniform in the nuclear positions.
The restrictions imposed on the physical 
parameters in \cite{LiebLoss2002} do, however, not allow for
all atomic numbers less than $Z_\np$.

First, we consider the semi-relativistic Pauli-Fierz operator.
The following proposition is a simple generalization of the bound
\eqref{maria} proven in \cite{MatteStockmeyer2009a} 
to the case of $N\in\NN$
electrons and $K\in\NN$ nuclei.

\begin{proposition}\label{prop-sb-srPF}
Assume that $\omega$ and $\V{G}$ fulfill Hypothesis~\ref{hyp-G}
and let $N,K\in\NN$, $e>0$, $\sZ=(Z_1,\ldots,Z_K)\in(0,2/\pi e^2]^K$,
and $\sR=\{\V{R}_1,\ldots,\V{R}_K\}\subset\RR^3$.
Then
\begin{equation}\label{manuela-1}
\sum_{i=1}^N|\DA^{(i)}|\,+\,\VC\,+\,\delta\,\Hf\,\grg\,
-C(\delta,N,\sZ,\sR,d_1)\,>\,-\infty\,,
\end{equation}
for every $\delta>0$ in the sense of quadratic forms on $\sD_N$.
\end{proposition}

\begin{proof}
In view of \eqref{maria} we only have to explain how to
localize the non-local kinetic energy terms. 
To begin with we recall the following bounds proven in
\cite[Lemmata~3.5 and~3.6]{MatteStockmeyer2009a}:
For every $\chi\in C^\infty(\RR^3_\V{x},[0,1])$, 
\begin{align}\label{tim1}
\|\,[\chi,\SA]\,\|\,\klg\,\|\nabla\chi\|_\infty\,,
\qquad
\big\|\DA\,\big[\chi\,,\,[\chi,\SA]\,\big]\big\|\,
\klg\,2\,\|\nabla\chi\|_\infty^2\,.
\end{align}
Now, let $\ball{r}{\V{z}}$ denote the open ball of radius $r>0$
centered at $\V{z}\in\RR^3$ in $\RR^3$.
We set $\vr:=\min\{|\V{R}_k-\V{R}_\ell|:\,k\not=\ell\}/2$ and 
pick a smooth partition of unity on $\RR^3$, $\{\chi_k\}_{k=0}^K$,
such that $\chi_k\equiv1$ on $\ball{\vr/2}{\V{R}_k}$ and
$\supp(\chi_k)\subset \ball{\vr}{\V{R}_k}$, for $k=1,\ldots,K$,
and such that $\sum_{k=0}^K\chi_k^2=1$.
Then we have the following IMS type localization formula,
\begin{equation}\label{IMS-DA1}
|\DA|\,=\,
\sum_{k=0}^K\Big\{\,\chi_k\,|\DA|\,\chi_k\,+\,\frac{1}{2}\,
\big[\chi_k\,,\,[\chi_k,|\DA|\,]\,\big]\,\Big\}
\end{equation}
on $\sD$, for every $i\in\{1,\ldots,N\}$.
A direct calculation shows that
\begin{align}\label{IMS-DA2}
\big[\chi_k\,,\,[\chi_k,|\DA|\,]\,\big]\,&=\,
2\,i\valpha\cdot(\nabla\chi_k)\,[\chi_k,\SA]+
\DA\,\big[\chi_k\,,\,[\chi_k,\SA]\,\big]
\end{align}
on $\sD$. By virtue of \eqref{C*} and \eqref{tim1} we thus get
\begin{equation}\label{IMS-Da3}
\big\|\,\big[\chi_k\,,\,[\chi_k,|\DA|\,]\,\big]\,\big\|\,
\klg\,4\,\|\nabla\chi\|_\infty^2\,,
\end{equation}
for all $k\in\{0,\ldots,K\}$.
Since we are able to localize the kinetic energy terms
and since, by the choice of the partition of unity,
the functions $\RR^3\ni\V{x}\mapsto
|\V{x}-\V{R}_k|^{-1}\,\chi_\ell^2(\V{x})$
are bounded, for $k\in\{1,\ldots,K\}$,
$\ell\in\{0,\ldots,K\}$, $k\not=\ell$,
the bound \eqref{manuela-1} is now an immediate 
consequence of \eqref{maria} (with $\delta$ replaced by $\delta/N$).
(Here we also make use of the fact that the hypotheses on $\V{G}$
are translation invariant.) 
\end{proof}

\smallskip

\noindent
Next, we turn to the no-pair operator discussed
in Section~\ref{sec-np}.
The semi-boundedness of the molecular $N$-electron no-pair operator
is essentially a consequence of the following
inequality 
\cite[Equation~(2.14)]{MatteStockmeyer2009a},
valid for all
$\omega$ and $\V{G}$ fulfilling Hypothesis~\ref{hyp-G},
$\gamma\in(0,2/(2/\pi+\pi/2))$, and $\delta>0$,
\begin{equation}\label{lb-np}
\PA\,(\DA^{(i)}-\gamma/|\V{x}|+\delta\,\Hf)\,\PA
\,\grg\,
\PA\,(c(\gamma)\,|\DO|-C)
\,\PA\,,
\end{equation}
in the sense
of quadratic forms on $\PA\,\sD$.
Here $C\equiv C(\delta,\gamma,d_{-1},d_0,d_1)\in(0,\infty)$
and $c(\gamma)\in(0,\infty)$ depends only on $\gamma$.

\begin{proposition}\label{prop-sb-np}
Assume that $\omega$ and $\V{G}$ fulfill Hypothesis~\ref{hyp-G}
and let $N,K\in\NN$, $e>0$,
$\sZ=(Z_1,\ldots,Z_K)\in(0,Z_\np)^K$, 
and $\sR=\{\V{R}_1,\ldots,\V{R}_K\}\subset\RR^3$,
where $Z_\np$ is defined in \eqref{def-Znp}. 
Then the quadratic form associated with
the operator $\wt{H}_\np$ defined in \eqref{def-Hnp'}
is semi-bounded below,
$$
\wt{H}_\np\,\grg\,-C(N,K,\sZ,\sR,d_{-1},d_1,d_5)\, >\,-\infty\,,
$$
in the sense of quadratic forms on $\sD_N$.
\end{proposition}

\begin{proof}
We again employ the parameter $\vr>0$ and the
partition of unity introduced in the
paragraph succeeding \eqref{tim1}.
Thanks to \cite[Lemma~3.4(ii)]{MatteStockmeyer2009a}
we know that $\PA$ maps $\dom(\DO\otimes\Hf^\nu)$
into itself, for every $\nu>0$.
The IMS localization formula thus yields
\begin{align*}
\Pa{i}\,v_i\,\Pa{i}&=
\sum_{k=0}^K
\Big\{\chi_k^{(i)}\,\Pa{i}\,v_i\,\Pa{i}\,\chi_k^{(i)}\,+\,\frac{1}{2}\,
\big[\chi_k^{(i)}\,,\,[\chi_k^{(i)}\,,\,\Pa{i}\,v_i\,\Pa{i}]\,\big]\Big\}
\end{align*}
on $\dom(\DO\otimes\id)$, where a superscript $(i)$ indicates
that $\chi_k=\chi_k^{(i)}$ depends on the variable $\V{x}_i$. 
Using $v_i\klg0$, we observe that
\begin{align}\nonumber
\big[\chi_k^{(i)}&\,,\,[\chi_k^{(i)}\,,\,\Pa{i}\,v_i\,\Pa{i}]\,\big]
\\
&=\,\nonumber
-2\,[\chi_k^{(i)}\,,\,\Pa{i}]\,v_i\,[\Pa{i}\,,\,\chi_k^{(i)}]
+2\,\Re\big\{\,\Pa{i}\,v_i\,\big[\chi_k^{(i)}\,,
\,[\chi_k^{(i)}\,,\,\Pa{i}]\,\big]\,\big\}
\\
&\grg\,\label{paola1}
2\,\Re\big\{\,\Pa{i}\,v_i\,\big[\chi_k^{(i)}\,,\,[\chi_k^{(i)}
\,,\,\Pa{i}]\,\big]\,\big\}
\,.
\end{align}
We recall the following estimate proven in
\cite[Lemma~3.6]{MatteStockmeyer2009a},
for every $\chi\in C^\infty(\RR^3_\V{x},[0,1])$,
$$
\big\|\tfrac{1}{|\V{x}|}\,\big[\chi\,,\,[\chi\,,\,\PA]\,\big]\,
\HT^{-1/2}\,\big\|\,\klg\,8^{3/2}\,\|\nabla\chi\|_\infty^2\,,
$$
where $\HT=\Hf+E$ with $E\grg1\vee(4d_1)^2$.
Together with \eqref{paola1} it implies 
\begin{align*}
\SPb{\Psi}{\big[\chi_k^{(i)}\,,\,[\chi_k^{(i)}\,,\,\Pa{i}\,
v_i\,\Pa{i}]\,\big]\,\Psi}
\grg-\delta\,\SPn{\Psi}{\HT\,\Psi}
-(8^3\,\|\nabla\chi_k\|_\infty^4/\delta)\,\|\Psi\|^2,
\end{align*}
for all $k\in\{0,\ldots, K\}$, $i\in\{1,\ldots,N\}$, $\delta>0$,
and $\Psi\in\dom(\DO\otimes\id)$.
Next, we pick cut-off functions, 
$\zeta_1,\ldots,\zeta_K\in C_0^\infty(\RR^3_\V{x},[0,1])$,
such that $\zeta_k=1$ in a neighborhood of $\V{R}_k$ and
$\supp(\zeta_k)\subset\ball{\vr/4}{\V{R}_k}$, for
$k\in\{1,\ldots,K\}$.
By construction, $\supp(\zeta_k)\cap\supp(\chi_\ell)=\varnothing$,
for all $k\in\{1,\ldots,K\}$ and $\ell\in\{0,\ldots,K\}$
with $k\not=\ell$.
Denoting $\ol{\zeta}_k:=1-\zeta_k$
and using the superscript $(i)$ to indicate that
$\zeta_k=\zeta_k^{(i)}$ is a function of the variable $\V{x}_i$,
we obtain
\begin{align}
\SPb{\Psi}{\chi_k^{(i)}\,\Pa{i}\,v_i\,\Pa{i}\,\chi_k^{(i)}\,\Psi}
\,&=\,\nonumber
-\SPB{\Psi}{\chi_k^{(i)}\,\Pa{i}\,
\frac{e^2\,Z_k}{|\V{x}_i-\V{R}_k|}\,\Pa{i}\,\chi_k^{(i)}
\,\Psi}
\\ \label{vica3}
&\quad-\sum_{{\ell=1\atop\ell\not=k}}^K
\SPB{\Psi}{\chi_k^{(i)}\,\Pa{i}\,
\frac{e^2\,Z_\ell\,\zeta_\ell^{(i)}}{|\V{x}_i-\V{R}_\ell|}\,\Pa{i}\,
\chi_k^{(i)}\,\Psi}
\\
&\quad-\sum_{{\ell=1\atop\ell\not=k}}^K\label{vica4}
\SPB{\Psi}{\chi_k^{(i)}\,\Pa{i}\,
\frac{e^2\,Z_\ell\,\ol{\zeta}_\ell^{(i)}}{|\V{x}_i-\V{R}_\ell|}
\,\Pa{i}\,\chi_k^{(i)}\,\Psi},
\end{align}
for all $\Psi\in\dom(\DO\otimes\id)$.
The operators appearing in the scalar products in \eqref{vica4}
are bounded by definition of $\ol{\zeta}_\ell$.
Their norms depend only on $\sR$ since $e^2\,Z_\ell<1$.
Furthermore, by virtue of Lemma~\ref{le-vica} below
the term in \eqref{vica3} is bounded from below 
by $-\delta\,\SPn{\Psi}{\Hf\,\Psi}-C_\delta\,\|\Psi\|^2$,
for all $\delta>0$ and some $C_\delta\equiv
C_\delta(\sR,d_1,d_4)\in(0,\infty)$;
see \eqref{vica2}.

Taking all the previous remarks into account,
using \eqref{IMS-DA1}--\eqref{IMS-Da3},
$w_{ij}\grg0$, $|\DA^{(i)}|\grg\Pa{i}\,\DA^{(i)}\,\Pa{i}$, and writing 
$$
\Hf\,=\,\frac{1}{N}\sum_{i=1}^N\sum_{k=0}^K\chi_k^{(i)}\,
(P_\V{A}^{+,(i)}+P_\V{A}^{-,(i)})\,\Hf\,\chi_k^{(i)}\,,
$$ 
we deduce that
\begin{align*}
\wt{H}_\np
\,&\grg
(1-3\delta)\,\PAN\,\Hf\,\PAN\,+\,(1-3\delta)\,\PANb\,\Hf\,\PANb
\\
&+\sum_{\sharp\in\{+,\bot\}}\sum_{k=0}^K\PANs\,\Big\{
\sum_{i=1}^N\chi_k^{(i)}\,
P_\V{A}^{+,(i)}\,\Big(\DA^{(i)}
-\frac{e^2\,Z_k}{|\V{x}_i-\V{R}_k|}
+\frac{\delta}{N}\,\Hf\Big)\,P_\V{A}^{+,(i)}\,\chi_k^{(i)}
\\
&+\,\frac{\delta}{N}\sum_{i=1}^N\Big(\,\chi_k^{(i)}\,
P_\V{A}^{-,(i)}\,\Hf\,P_\V{A}^{-,(i)}\,\chi_k^{(i)}
+\sum_{\flat=\pm}\chi_k^{(i)}\Paf{i}\,
[\Paf{i},\Hf]\,\chi_k^{(i)}\,\Big)\,\Big\}\,\PANs
\\
&-\,\const(N,\sR,d_1,d_4)
\end{align*}
on $\sD_N$,
for every $\delta>0$.
Thanks to Corollary~\ref{cor-clara1} (with $\ve=0$) we know that 
$[\Paf{i},\Hf]\,\HT^{-1/2}$ extends to an element of $\LO(\HR_N)$
whose norm is bounded by some constant depending only
on $d_1$ and $d_5$,
whence
\begin{align*}
\frac{\delta}{N}\sum_{i=1}^N\sum_{k=0}^K&\SPb{\chi_k^{(i)}\,\PANs\,\Psi}{
\Paf{i}\,[\Paf{i},\Hf]\,\chi_k^{(i)}\,\PANs\,\Psi}
\\
&\grg\,
-(\delta/2)\,\|\HT^{1/2}\,\PANs\,\Psi\|^2
-(\delta/2)\,\big\|[\Paf{i},\Hf]\HT^{-1/2}\big\|^2
\,\|\Psi\|^2,
\end{align*}
for every $\Psi\in\sD_N$, $\sharp\in\{+,\bot\}$, 
and $\flat=\pm$. 
For a sufficiently small choice of $\delta>0$,
the assertion now follows from the semi-boundedness
of $\Pa{i}\,(\DA^{(i)}-e^2\,Z_k/|\V{x}_i-\V{R}_k|+(\delta/N)\,\Hf)\,\Pa{i}$ 
ensured by
\eqref{lb-np} and the condition $Z_k<Z_\np$.
\end{proof}

\begin{lemma}\label{le-vica}
Let $\zeta\in C_0^\infty(\RR^3,[0,1])$, $\chi\in C^\infty(\RR^3,[0,1])$,
such that $0\in\supp(\zeta)$ and
$\supp(\zeta)\cap\supp(\chi)=\varnothing$.
Set $\HT:=\Hf+E$, where $E\grg k_1\vee d_1^2$.
Then 
\begin{align}\label{vica1}
\big\|\,\DA\,\Hf^{1/2}\,\zeta\,\PA\,\chi\,\HT^{-1/2}\,\big\|\,&\klg\,
C(\zeta,\chi,d_1,d_4)\,,
\\ \label{vica2}
\big\|\,\tfrac{\zeta}{|\V{x}|}\,\PA\,\chi\,\HT^{-1/2}\,\big\|\,&\klg\,
C'(\zeta,\chi,d_1,d_4)\,.
\end{align}
\end{lemma}

\begin{proof}
We pick some $\wt{\chi}\in C^\infty(\RR^3,[0,1])$
such that $\supp(\wt{\chi})\cap\supp(\zeta)=\varnothing$
and $\wt{\chi}\equiv1$ on $\supp(\nabla\chi)$.
Using $\zeta\,\chi=0=\zeta\,\wt{\chi}$
we infer that, for all $\vp,\psi\in\sD$,
\begin{align*}
\big|\SPb{&\DA\,\vp}{\Hf^{1/2}\,\zeta\,\PA\,\chi\,\HT^{-1/2}\,\psi}\big|
\,=\,
\big|\SPb{\DA\,\vp}{\Hf^{1/2}\,\zeta\,[\PA\,,\,\chi]\,\HT^{-1/2}\,\psi}\big|
\\
&\klg
\int_\RR\Big|\SPB{\DA\,\vp}{\Hf^{1/2}\,\zeta\,
[\RA{iy}\,,\,\wt{\chi}]\,
i\valpha\cdot\nabla\chi\,\RA{iy}\,\HT^{-1/2}\,\psi}\Big|
\frac{dy}{2\pi}
\\
&=\,
\int_\RR\Big|\SPB{\DA\,\vp}{\Hf^{1/2}\,\zeta\,
\RA{iy}\,i\valpha\cdot\nabla\wt{\chi}\,
\RA{iy}\,i\valpha\cdot\nabla\chi\,\RA{iy}\,\HT^{-1/2}\,\psi}\Big|
\frac{dy}{2\pi}
\\
&=\,
\int_\RR\Big|\SPB{\zeta\,\DA\,\vp}{
\RA{iy}\,\Upsilon_{0,1/2}(iy)\,i\valpha\cdot\nabla\wt{\chi}\,
\RA{iy}\,\Upsilon_{0,1/2}(iy)\,\times
\\
&\qquad\qquad\qquad\qquad\qquad\qquad
\times\,i\valpha\cdot\nabla\chi\,\RA{iy}\,\Upsilon_{0,1/2}(iy)\,\psi}\Big|
\frac{dy}{2\pi}\,.
\end{align*}
In the last step we repeatedly applied \eqref{eva2001}.
Commuting $\zeta$ and $\DA$ and using
$\|\DA\,\RA{iy}\|\klg1$, $\|\RA{iy}\|^2\klg(1+y^2)^{-1}$, and the fact that 
$\|\Upsilon_{0,1/2}(iy)\|$ is uniformly bounded in $y\in\RR$, 
we readily deduce that
$$
\big|\SPb{\DA\,\vp}{\Hf^{1/2}\,\zeta\,\PA\,\chi\,\HT^{-1/2}\,\psi}\big|
\,\klg\,C(\zeta,\chi,\wt{\chi},d_1,d_4)\,\|\vp\|\,\|\psi\|\,,
$$
which implies \eqref{vica1}. The bound
\eqref{vica2} follows from \eqref{vica1} and the inequality
$$
\|\,|\V{x}|^{-1}\,\vp\,\|^2\,\klg\,
4\,\|\,\DA\,\vp\,\|^2+4\,\|\HT^{1/2}\,\vp\|\,,
\qquad \vp\in\dom(\DO\otimes\Hf^{1/2})\,,
$$
which is a simple consequence of standard arguments
(see, e.g., \cite[Equation~(4.7)]{MatteStockmeyer2009a}).
\end{proof}


\bigskip

\noindent
{\bf Acknowledgement.}
It is a pleasure to thank Martin K\"onenberg and 
Edgardo Stockmeyer for interesting discussions and helpful remarks.

\def\cprime{$'$} \def\cprime{$'$} \def\cprime{$'$} \def\cprime{$'$}
  \def\cprime{$'$}

\end{document}